%% file: c-itqr.tex
\documentclass[11pt]{article}
\usepackage{amsmath,amssymb,epsfig}
\usepackage[frenchb,english]{babel}
\usepackage{makeidx}
\usepackage{wasysym}
\include{MWrite}

\include{MEnvA}

\include{MBibli}

\include{fmacros_a}

\input xypic

\title{Proving that a Tree Language is not First-Order Definable}

\author{Martin Beaudry\thanks{D\'epartement 
d'informatique, Universit\'e de Sherbrooke, 
Sherbrooke (Qc) Canada, J1K 2R1, martin.beaudry@usherbrooke.ca.
Work supported by NSERC of Canada.}}

 \date{}           

\begin{document}

\maketitle


\newcommand{\Francois}{Fran\c{c}ois}



 \begin{abstract}
 
We explore from an algebraic viewpoint the properties of the tree languages
definable with a first-order formula involving the ancestor  predicate,
using the description of these languages
as those recognized by
iterated block products of forest algebras
defined from threshold-$\tau$, period-$\pi$ counter monoids.
Ehrenfeucht-Fra\"{\i}ss\'e games, i.e.
proofs of non-definability,
are infinite sequences of sets of forests,
one for each level of the hierarchy 
of quantification levels that defines the corresponding
variety of languages.
A proof is recursive when the forests at a given level
are built  by inserting forests from the previous level  at the ports of
a suitable set of multicontexts.
We show  that a recursive proof exists for the
syntactic algebra  of every non-definable language.
We also investigate certain types of uniform recursive proofs.
For this purpose, we define from a forest algebra 
an ``algebra of mappings'' and an ``extended algebra'',
which we also use to redefine the notion of aperiodicity
in a way that generalizes the existing ones. 
 \end{abstract}

\normalsize


\setlength {\baselineskip} {16pt}

\section{Introduction}

Words and trees are used almost universally in Computer Science, and
logical formalisms are among the most convenient tools for specifying
these objects, or sets thereof.
Automata constitute another class of tools, procedural in nature, 
widely used to define languages; 
underlying this formalism is a rich algebraic theory, through which
further tools from other areas of Mathematics can be used to better understand
the properties of word and tree languages.
Moreover, 
the most significant classes of these languages happen to have
descriptions in several formalisms. For example,
the regular word languages are exactly those that are
recognized by finite automata and monoids, and those that 
are definable by monadic second-order formulas with a
unary predicate for each letter and a ``left-of'' binary positional
predicate.
Similarly the regular tree languages are simultaneously
those that are definable 
by monadic second-order formulas with two
positional predicates (``ancestor-of'' and ``next-sibling-of'')
and those that are recognized by finite tree automata,
as well as various sorts of  algebraic structures
(e.g. finite term algebras, finite forest algebras, finitary preclones). 
\newline
In the same vein, the word languages definable with
first-order logical formulas with the ``left-of'' predicate
have several algebraic and combinatorial descriptions,
see \cite{pi97,sc65,krrh65,rhti89,th82}.
In particular, these languages are
precisely those whose syntactic monoid is
aperiodic; 
thanks to this property, 
whether a regular word language is
first-order definable
can be determined with 
a straightforward
algorithm.
In the world of tree languages, 
however, none of the definitions for
aperiodicity of the
syntactic algebra tried so far has managed
to characterize precisely the
first-order definable languages
\cite{he91,bostwa09},
and the techniques invented to show that
certain important subclasses of these languages
are indeed decidable \cite{bese09,bosest08,bose08,plse11,plse15}\
did not seem to extend to the whole class.
%
%
\newline
Forest algebras combine two monoids (horizontal
and vertical) in a way which makes it easy  for  researchers to 
apply techniques from the theory of monoids and word languages.
An encouraging harvest of results has already been obtained with this tool \cite{bowa08,bosest08,bostwa09}.
In this paper, we look at a counterpart, in the world of trees and forests, to the description of the
aperiodic monoids as
the variety of monoids generated by iterated block products of semilattices \cite{rhti89}.
This is a
description of the first-order definable languages,
developed in terms of  the variety of
finitary preclones \cite{eswe10}\ generated by iterated
block products of preclones that count occurrences
of node labels, regardless of the actual tree
structure. 
Intuitively, a block product
works as the 
combination of two tree automata
where at every node $x$, the second automaton,
besides the label of $x$, also reads
the current state reached by the first automaton
after reading the subtree rooted 
at $x$ (``below'' $x$)
and the outcome of the
processing, by the first automaton,
of the context of $x$ within the tree (``above'' $x$).
\newline
The description of first-order definable languages
developed in \cite{hath87,mora03}\ suggests that
the $\equiv_{\tau,\pi}$ threshold-$\tau$, period-$\pi$ numerical congruences
are a fundamental feature in the combinatorics of first-order definable
languages. 
Consistently with this we use the same kind of counting
and the corresponding quotient, counter monoids
$\N_{\tau,\pi} = \N / {\equiv_{\tau,\pi}}$
(in these notations, 
the Boolean OR  $U_1$ and the cyclic group $\Z_\pi$
are respectively $\N_{1,1}$ and $\N_{0,\pi}$).
In our formalism, 
we denote by $\ODalg_M$ the  one-dimensional algebra where $M$ is the horizontal monoid,
by
$\bicrochet{\N_{\tau,\pi}} $ the variety of forest algebras generated by
$\ODalg_{\N_{\tau,\pi}}$,
by
$ \mathbf{**}^n \bicrochet{\N_{\tau,\pi}} $ the variety generated by iterated block 
products of $n$ algebras from $\ODalg_{\N_{\tau,\pi}}$,
and by $ \mathbf{**} \bicrochet{\N_{\tau,\pi}} $ the closure of these varieties
over joint.
Let
$\mathsf{FO}[\Panc]  $  and $\mathsf{FOMod_\pi}[\Panc]$ denote
both the class of forest languages definable by first-order
formulas built with  the
``ancestor'' positional predicate and the usual
quantifiers only (for $\mathsf{FO}[\Panc]  $) or the same with the
$\exists^\pi_i$, $0 \le i < \pi$, modular quantifiers,
and the varieties  generated by their syntactic algebras.
Using the formalism of finitary preclones they introduced in
\cite{eswe05},  by Es\'{\i}k and Weil 
have established in \cite{eswe10}\ the correspondences
$\mathsf{FO}[\Panc]  =   \mathbf{**}\bicrochet{\N_{1,1}} $
and
$\mathsf{FOMod_\pi}[\Panc]  =  \mathbf{**} \bicrochet{\N_{1,\pi}} $.
\newline
We explore two ways of defining from $\calG = (G,W)$
another algebra,  where multicontexts are the underlying objects.
The algebra of mappings $\calG_\#$, and the
``multivertical monoid'' of \calG\  derived from it 
enable us to
define  notions of pumping and aperiodicity that generalize two
of the known necessary conditions for membership in
$\mathsf{FO}[\Panc]$, namely aperiodicity of the vertical monoid and the
``absence of vertical confusion on uniform multicontexts''
defined in \cite{bostwa09}.
The extended algebra $\calG_\% = (G_\%,W_\%)$, where
$G_\%$
is the powerset of $G$, makes explicit some properties of \calG\ 
that are not directly visible in  \calG\ or  $\calG_\#$.
An example is described in Section \ref{sec:potthoff}:
 this is a language whose syntactic algebra
has aperiodic vertical and multivertical monoids,
but where $W_\%$ is divided by the group $\Z_2$.
The language is defined with a $\mathsf{FOMod_2}[\Panc]$ formula 
that, among other things, counts the parity of the length of 
certain node-to-leaf paths;
a pair of elements of  $W_\%$ does precisely this counting.
\newline
An algebra \calG\ lies outside of the variety
$ \mathbf{**} \bicrochet{\N_{\tau,\pi}} $
if, and only if  there exists 
an infinite sequence of sets $\calS^{(n)}$, one for each $n \ge 1$, of forests belonging
to different languages recognized by \calG, such that
the elements of
$\calS^{(n)}$  cannot be told apart by any forest algebra in $ \mathbf{**}^n  \bicrochet{\N_{\tau,\pi}} $.
This sequence is usually  described through an Ehrenfeucht-Fra\"{\i}ss\'e game.
We call such a sequence a proof of non-membership in $ \mathbf{**} \bicrochet{\N_{\tau,\pi}}  $.
An Ehrenfeucht-Fra\"{\i}ss\'e game actually builds a recursive proof,
where each forest of
$\calS^{(n+1)}$  is  built by inserting
copies of the elements of $\calS^{(n)}$ at the ports of 
the corresponding
element of a set $\calM^{(n)}$ of multicontexts.
Such a proof can be specified with an infinite sequence of 
such sets $\calM^{(n+1)}$;
 we denote by
$\mathbf{RC}(\calG {\star} \calH^n_{\tau,\pi})$
the proposition that states the existence  a $\calM^{(n+1)}$
that has the required properties;
RC stands for ``recursive construction''.
We prove  that 
$\calG \not\in \mathbf{**} \bicrochet{\N_{\tau,\pi}} $
if, and only if $\mathbf{RC}(\calG  {\star} \calH^n_{\tau,\pi})$ holds
for every $n \ge 1$,
that is, every algebra that is not first-order has a recursive proof
of non-membership. %
Next, we observe that in the existing proofs, the circuit
$\calM^{(n+1)}$ is either identical to
$\calM^{(n)}$, in which case  each forest of
$\calS^{(n)}$ is  built from copies of the same, finite set
of multicontexts (``proof-by-copy''), or 
$\calM^{(n+1)}$ is obtained by pumping a starting set
of multicontexts (``proof-by-pumping''). 
The questions of the existence of a
proof-by-copy
and of a restricted form of proof-by-pumping
are both recursively enumerable.
\newline
Section  \ref{sec:defback}\
contains  background 
on forests, multicontexts and circuits,
and on forest algebras 
and the varieties $ {{\mathbf{**}}}\bicrochet{ \N_{\tau,\pi}}$. 
In Section \ref{sec:multicon}, 
we define the algebras $\calG_\#$ and $\calG_\%$.
and the related notions of
pumping and aperiodicity.
%
%
In Section \ref{sec:main}, we prove  that an algebra is outside 
${{\mathbf{**}}} \bicrochet{ \N_{\tau,\pi}} $ if, and only if
this can be asserted with
a recursive proof; 
we then explore  the notions of
proof-by-copy and proof-by-pumping.
In Section \ref{sec:exaproof}, we 
discuss  in our formalism
some typical examples of
non-membership proofs.
We conclude with some comments and open questions.

\section{Definitions and Background}  \label{sec:defback}

\subsection{Forests, Multicontexts, and Circuits}  \label{sec:FMC}

We consistently work with a finite alphabet
$A$, which we assume to always contain a neutral letter $e$,  such that for every forest homomorphism $\gamma$,
$\gamma(e\square) $ is the identity mapping. Let $B$ be another alphabet, disjoint from $A$.
A \emph{multicontext} $m$ over $(A,B)$ is a sequence of trees in which 
a subset of the leaves consists of \emph{ports}, where 
every non-port node $y$
carries a label  $\lambda(m,y) \in A$.
We denote by $nodes(m)$ the set of all nodes in $m$,
by $ports(m)$ the set of its ports and by $interior(m)$ the set of the non-port nodes.
We work with multicontext where each port either has  a label $\nu(x) \in B$,
or has several labels, each
specified as a mapping from  $ports(m)$ to a set that is disjoint from $A$.
A \emph{forest} is  a multicontext without ports;
a \emph{context} in the usual sense is a forest with a unique port that
carries the special label $\square \not\in (A \cup B)$,
called its $\square$-port. Throughout the paper,
this port is considered apart from the others.
%
Given  $x \in nodes(m)$,
we denote by $\Delta(m,x) $ the multicontext 
consisting of all subtrees of $m$ rooted at the sons of $x$.
The subtree rooted at $x$, i.e.  $\Delta(m,x)$ plus the node $x$,
is denoted $\Delta^+(m,x)$.
The
context of $x$ within $m$, with notation $\nabla(m,x) $,
is built from $m$ and $x$ by replacing
$\Delta^+(m,x)$ with a  $\square$-port.
The ancestors of  this port constitute the \emph{trunk} of the context.
If we deal with a set of multicontexts $M$ instead of an individual $m$, we 
use the notations $nodes(M)$, $\lambda(M,x)$, $\Delta(M,x)$, etc.
The sets of all forests,  and contexts over $A $ are 
respectively denoted 
\Snat\ and \Svert. We use the notations $\BbbM_{A,B}$ and $\BbbC_{A,B}$,
respectively,
for the  set of all multicontexts over $(A,B)$ for the 
set of all multicontexts with a $\square$-port
(the contexts-in-multicontexts, so to speak).
We use the
standard representation for individual forests of multicontexts, 
where nodes are listed in preorder and where
concatenation and $+$  represent the father-son relation and
horizontal addition, respectively. For example, $a(b+c)$ 
is a tree with a root labelled $a$ and two sons labelled $b$ and $c$,
while
$ab+c$ is a forest of two trees, where nodes $a$ and $c$ are roots,
and nodes $b$ and $c$ are leaves.
\newline 
Inserting $t$ in a context $s$ consists in replacing the $\square$-port of $s$ 
with a copy of $t$; the resulting forest is
denoted $s {\cdot} t$, or $st$.
Insertion in multicontexts is done
here either on a wholesale basis, i.e. something
is inserted at every port,
or
on a selective basis, when insertion occurs at a 
pre-specified set of ports.
The latter method is defined  in Section \ref{sec:multicon};
the former is associated with circuits and the construction of witnesses,
as follows.
\newline
Let $A$, $B$ and $C$ be three sets.
A \emph{circuit} over  $(A,B,C)$
is a set $M$ with
an element $m_c$ for every $c \in C$;
this component is a multicontext $m \in \BbbM_{A,B}$.
We can regard $M$ as having an input  wire
for every $b \in B$,  an output wire for every $c \in C$,
and $m_c$ is the result of
unraveling into a tree all those nodes of $M$ from which
the output wire $c$ is accessible.
A set $S $ of forests over $(A,B)$
is defined similarly, with an element  $ s_b \in  \Snat$
for every $ b \in B$. The insertion of $S$ in 
a circuit $M$ over $(A,B,C)$ consists in
inserting a copy of the forest $ s_{\nu(x)}$
at every $x \in ports(M)$;
the result is a set of forests over $(A,C)$,
denoted $M {\cdot} S$.
If $M$ and $M'$ are  circuits over  $(A,B,B)$ 
then inserting $M'$ in $M$
builds a circuit $M {\cdot} M'$ over $(A,B,B)$.
It can be verified, using standard methods, that this operation is associative.

\subsection{Forest algebras}   \label{sec:forestalg}


The reader is assumed to be knowledgeable with the notions of
semigroups and monoids, and their relations with regular languages,
word congruences and monoid homomorphisms (see \cite{pi84,pi97}).
Two types of notations are used for
 the monoids discussed in the article. There is an additive, or ``horizontal'' notation where
 the identity and operation are denoted $0$ and  $+$, respectively, although this does in no way 
imply that the latter is commutative. In the multiplicative, or ``vertical'' notation, the neutral element is denoted 
$\varepsilon$
and the operation is written with $\cdot$ or by concatenation of the arguments.
\newline
A \emph{ transformation } of a set $S$ is a mapping $S \rightarrow S$, i.e. an element
of the monoid $S^S$.
A \emph{ translation } in a monoid $M$ (with the additive notation) is  a mapping
$[u + \varepsilon + v] : M \rightarrow M$, where $u,v \in M$,
defined by $s \mapsto u+s+v$. If  $M$ is commutative, then
the translations are of the form  $[u + \varepsilon ]  $.
The set $\calM(M) = \{ [u + \varepsilon + v]\ |\ u,v \in M \}$ with the composition of functions
is the \emph{translation monoid} of $M$. 
%

\begin{definition} \label{def:algfor}
A \emph{forest algebra} is
 a pair $\calH = (H,V)$ where $(H,+)$ is a monoid and $(V,\cdot)$ 
 is a submonoid of $H^H$ which contains $\calM(H)$.
 \end{definition}
 
Monoids $H$ and $V$ are the
 horizontal and  vertical monoids of \calH, respectively.
Because $V$ is a submonoid of $H^H$, its  action on $H$ 
is faithful.
Forest algebras were introduced in \cite{bowa08}\ as pairs of abstract monoids;
in that case, faithfulness has to be specified in the definition.
\newline
A \emph{forest algebra homomorphism} 
from $\calH = (H,V)$ to $\calG = (G,W)$
is a pair of mappings $\alpha = (\alpha_H,\alpha_V)$ where
 $\alpha_H $ and $\alpha_V $   are monoid homomorphisms $H \rightarrow G$
 and $V \rightarrow W$, respectively, and such that 
 $\alpha_H \circ w = \alpha_V(w) \circ \alpha_H$ for every $w \in V$. 
The \emph{free forest algebra} over $A$ is
$\Sfree =  (\Snat,\Sboxvert)$; since it is generated from 
$\{a\square\ |\ a \in A\}$,
 a homomorphism $\alpha : \Sfree \rightarrow \calH$
is completely  specified once \calH\ and every $\alpha_V(a\square)$, $a \in A$, are known.
A \emph{forest congruence} in \Sfree\ is a pair of
 equivalence relations, both denoted by $\approx$, such that in \Snat\ 
$s \approx t$ iff $ps \approx pt$ for every context $p $,
and
 in \Sboxvert\
$p \approx q$ iff $pt \approx qt$ for every forest $t $.
A congruence $\approx$ \emph{refines} another congruence 
$\simeq$ over the same domain,
when $x \approx x' \Rightarrow x \simeq x'$ for all $x,x'$.
A homomorphism $\alpha$ defines  its nuclear congruence:
 $s \approx_\alpha t \Leftrightarrow \alpha(s)=\alpha(t)$,
 and conversely a congruence $\approx$ defines a  homomorphism from \Sfree\ to $\Sfree/\approx$.
A set $L \subset \Snat$ is
\emph{recognized by}  \calH\ if there exist
a  homomorphism $\alpha : \Sfree \rightarrow \calH$
and a subset $F \subset H$
such that for all $t \in \Snat$,
$t \in L \Leftrightarrow \alpha(t) \in F$. A context language $K \subset \Sboxvert$
is recognized in the same way, with an accepting set $P \subset V$.
The syntactic congruences of these languages are refined by  
 $\approx_\alpha$.
\newline
A \emph{variety of forest algebras} is a class of finite forest algebras closed under finite direct
product and division.
Given forest algebras
$(G,W)$ and $(H,V)$, we say that \calG\ is a subalgebra of \calH\ iff $G \subseteq H$ and $W \subseteq V$,
and that it divides \calH, with notation $\calG \prec \calH$, if it is the homomorphic image of a subalgebra
of \calH.
 A \emph{variety of forest languages} is formally defined
as a mapping $\mathsf{W}$  such that, for every alphabet $A$, $\mathsf{W}(A)$ 
is closed under finite boolean operations, inverse homomorphism of free algebras
and context quotients.
With  $L $ a language and $p \in \Sboxvert$, the context quotient of $L$ by $p$ is the
set $p^{-1}L = \{\ s \ |\ ps \in L\} $; a
forest  algebra which recognizes $L$ also recognizes $p^{-1}L$.
The lattices of varieties of forest algebras and of varieties of forest languages are isomorphic \cite{bostwa12}.
%
%
%
%
\newline
Let $\calG = (G,W)$ be a forest algebra and let $g,h \in G$.
An element $h$ is \emph{accessible} from $g$
when $ g = wh$ for some $w \in W$.
A set is
\emph{strongly connected} when its elements are mutually accessible;
a \emph{strongly connected component} of $G$ is a subset that is maximal for this property.
Let  $K$  be such a set: we define from it the set 
$W^{-1}K =  \{ g \in G : \exists\, w \in W,\, wg \in K \}$
of all elements from which $K$ is accessible, and 
its complement $I(K)$, which is an ideal,
that is,
a subset of $G$ closed under the
action of every element of $W$. 
%
Let $K \subset G$.
The \emph{leaf-completion} of a multicontext $m$ through 
a mapping $\chi : ports(m) \rightarrow K$
is the forest 
$\Brvchi(m) \in \BbbH_{A \cup K}$,
obtained  by
labeling every port $x$  with $\chi(x)$.
Consistently with this, 
the  \emph{leaf-extension} of 
a homomorphism $\gamma : \Sfree \rightarrow \calG$
to $\BbbH_{A \cup K}$
is built by defining $\gamma(t_k) = k$,
for every $h \in K$ and one-node tree
$t_k$ with label $k$.
Then 
$\gamma \circ \Brvchi(m)$ is
the image by $\gamma$ of the leaf-completion of  $m$ through 
$\chi $.
%


\subsection{Block product congruences}   \label{sec:block}

It is  known that the equivalence relation over \Sfree\
where every class consists in all forests that model the same
set of formulas of quantifier depth $n$, is a forest congruence.
A generalized version of this congruence is
$\approx^n_{\tau,\pi}$, where $n,\tau,\pi \ge 1$
are integers,
defined as follows:
\begin{listedense}
\item 
it is built around the threshold-$\tau$, period-$\pi$ counting
congruence over \N, defined by 
\[
p \equiv_{\tau,\pi} q \ \Leftrightarrow\ p=q \vee \left( p \ge \tau \wedge q \ge \tau \wedge 
(p-q) \equiv 0 (\text{mod}\ \pi)  \right)\, ;
\]
the quotient monoid $\N / {\equiv_{\tau,\pi} }$ is denoted $\N_{\tau,\pi}$;
\item 
given $s,t \in \Snat$, we have
$s \approx^1_{\tau,\pi} t$ if, and only if, for every $a \in A$,
the number of nodes with label $a$ in $s$ and in $t$
are congruent under $\equiv_{\tau,\pi}$;
the quotient algebra $\Sfree / {\approx^1_{\tau,\pi} }$ is denoted
$\calH^1_{\tau,\pi}$ ; the corresponding surjective homomorphism
is $\alpha^1_{\tau,\pi} : \Sfree \rightarrow \calH^1_{\tau,\pi} $;
\item for $n \ge 1$, 
given that $\approx^n_{\tau,\pi} $ and the quotient algebra
$\calH^n_{\tau,\pi} = (H^n_{\tau,\pi}, V^n_{\tau,\pi})$
are already known,
we define a relabeling operation
$s \mapsto s^{\alpha^n_{\tau,\pi}} $ which consists in 
replacing, at every node $x$ of $s$, the label $\lambda(s,x)  \in A$
with the triple 
\[
\lambda(s^{\alpha^n_{\tau,\pi}},x) \ =\ 
\crochet{ \lambda(s,x) , \alpha^n_{\tau,\pi} (\Delta(s,x)) ,  \alpha^n_{\tau,\pi} (\nabla(s,x)) } ;
\]
this defines the relabeling alphabet 
$D^{\alpha^n_{\tau,\pi}} = A \times H^n_{\tau,\pi}, \times V^n_{\tau,\pi}$;
the same is done in a context $t \in \Svert$; however, the new label of $x$
is different depending on whether $x$ is on the trunk,
so that
$t^{\alpha^n_{\tau,\pi}} $ is a context over 
$E^{\alpha^n_{\tau,\pi}} = (A_{\text{off}} \cup A_{\text{trunk}}) \times H^n_{\tau,\pi}, \times V^n_{\tau,\pi}$,
where $A_{\text{off}} $ and $A_{\text{trunk}}$ are disjoint copies of $A$;
\item for $n \ge 1$ and $s,t \in \Snat$, we have
$s \approx^{n+1}_{\tau,\pi} t$ if, and only if $s \approx^n_{\tau,\pi} t$ 
and 
$s^{\alpha^n_{\tau,\pi}}  \approx^1_{\tau,\pi} t^{\alpha^n_{\tau,\pi}} $. 
\end{listedense}
Example: 
we have $(a + \square) \approx^1_{\tau,\pi} a\square$ and
$(a + \square) \not\approx^2_{\tau,\pi} a\square$, which illustrates
the distinction between trunk and off-trunk nodes.
\newline
The quotient  algebra $\Sfree / {\approx^1_{\tau,\pi} }$ is isomorphic
to the one-dimensional algebra $\ODalg_{\N_{\tau,\pi}^{|A|}}$.
A {one-dimensional} {forest algebra}\footnote{Also called flat algebras in previous works on
the topic: the homomorphic image of a forest is the image in a monoid 
of a ``flattened'',
``one-dimensional'' version of the forest. The wording is also a
reference to the notion
that an 
algebra  $\ODalg_M {\Blockprod} \ODalg_N$ recognizes 
forest languages that are ``more two-dimensional'' than 
those recognized by $\ODalg_M$.} 
is a pair $(M,\calM(M))$ such that for every homomorphism
$\gamma : \Sfree \rightarrow (M,\calM(M))$ and every $a \in A$, there exists $m \in M$ such that
$\gamma(a\square) = [m + \varepsilon ]$. In such an algebra,
the homomorphic image of a forest $t \in \Snat$ is independent of its structure, that is, the algebra only
considers the  string $\eta(t)$ of its node labels, given in a predetermined order (e.g. in preorder).
Therefore,  $(M,\calM(M))$ associates to $\gamma$ a monoid homomorphism
$\hat{\gamma} : A^* \rightarrow M$, such that $\gamma_H = \hat{\gamma} \circ \eta $.
We denote by $\ODalg_M$ the (unique)
one-dimensional algebra built from   $M$. 
\newline
The congruences $\approx^n_{\tau,\pi}$ can be defined algebraically, as follows.
Let $\crochet{\N_{\tau,\pi}}$ 
and
${{\mathbf{**}}^1}\bicrochet{\N_{\tau,\pi}}$
denote respectively the variety of monoids generated by $\N_{\tau,\pi}$
and the variety of forest algebras generated by the algebras
$\ODalg_{M}$ where  $M \in \crochet{\N_{\tau,\pi}}$.
Then for every language $L \subseteq \Snat$,
its syntactic forest algebra $\calG(L)$ belongs to ${{\mathbf{**}}^1}\bicrochet{\N_{\tau,\pi}}$
if, and only if $\approx^1_{\tau,\pi}$ refines its syntactic congruence,
or equivalently, iff $\calG(L)$ divides $\calH^1_{\tau,\pi}$.
Next, every algebra $\calH^{n+1}_{\tau,\pi} = \Snat / {\approx^{n+1}_{\tau,\pi}}$ is a block product
$\calH^{n}_{\tau,\pi} \Blockprod \ODalg_{\N_{\tau,\pi}^{N}}$, with
$N \ge |D^{\alpha^n_{\tau,\pi}} \cup   E^{\alpha^n_{\tau,\pi}} |$.
We use ${{\mathbf{**}}^{n+1}}\bicrochet{\N_{\tau,\pi}}$ to denote the variety
generated by block products of the form $\calG \Blockprod \ODalg_{\N_{\tau,\pi}^{N}}$
with $\calG \in {{\mathbf{**}}^{n}}\bicrochet{\N_{\tau,\pi}}$.
Finally, ${{\mathbf{**}}}\bicrochet{\N_{\tau,\pi}} = \bigvee_{n \ge 1}  {{\mathbf{**}}^{n}}\bicrochet{\N_{\tau,\pi}}$.
We will make abundant use of the following.
%
%
\begin{proposition}  \label{prp:basic}
The following statements on a  finite-index congruence $\simeq$ over \Sfree\
are equivalent:
 $\Sfree / {\simeq} \in {{\mathbf{**}}^n}\bicrochet{\N_{\tau,\pi}}$;
  $\Sfree / {\simeq} \prec \calH^n_{\tau,\pi} $;  the  congruence
 $\approx^n_{\tau,\pi} $ refines $ \simeq$.
 \hfill $\square$
 \end{proposition}

Let $\mathsf{FO}[\Panc]$ denote the variety of all forest languages
definable with first-order logic formulas with the $\forall$ and $\exists$
quantifiers and the `ancestor' positional predicate; for $q \ge 2$,
let $\mathsf{FOMod_q}[\Panc]$ denote the variety defined in terms of the same
sort of formulas, where now the $\exists^q_i$, $0 \le i < q$, modular quantifiers are
also allowed.
It was proved in
\cite{eswe10}\ that the syntactic preclones of the languages in $\mathsf{FO}[\Panc]$
generate the same variety as the iterated block products of
preclones defined in terms of counting under threshold one (i.e. the monoid $U_1$);
adding to the generating preclones those 
defined with counting under the congruence $\equiv_{0,\pi}$
yields a characterization  for $\mathsf{FOMod_\pi}[\Panc]$.
%
It can be verified that these equivalences
translate into
$\mathsf{FO}[\Panc] = \bigvee_{\tau \ge 1}  {{\mathbf{**}}}\bicrochet{\N_{\tau,1}}$
and
$\mathsf{FOMod_\pi}[\Panc] = \bigvee_{\tau \ge 1}  {{\mathbf{**}}}\bicrochet{\N_{\tau,\pi}}$.
\newline
\emph{Remark}. Actually, 
$\mathsf{FO}[\Panc] =  {{\mathbf{**}}}\bicrochet{\N_{1,1}}$, where
$\N_{1,1} = U_1  $ is the Boolean OR monoid, and
similarly
$\mathsf{FOMod_\pi}[\Panc] =  {{\mathbf{**}}}\bicrochet{\N_{1,\pi}}$, 
so that working in terms of nontrivial thresholds $\tau$ is not mandatory.
However, doing so makes it possible to follow
more closely the counting-under-threshold that seems to be inherent to the 
construction of proofs of non-membership in $\mathsf{FO}[\Panc]$,
and is reminiscent to the description of the
first-order definable forest 
languages developed in \cite{hath87,mora03}.
Note that a characterization $\mathsf{Mod_\pi}[\Panc] =  {{\mathbf{**}}}\bicrochet{\Z_\pi} $
also exists,
where $\Z_\pi$ is the cyclic group of order $\pi$;
we put aside this special case in the current version of this paper.


\section{Algebras for Multicontexts}  \label{sec:multicon}


%
Forest algebras were designed as tools to handle trees, forests, and contexts over $A$.
Dealing with multicontexts over $(A,B)$ as we do in this article demands that 
a suitable algebraic structure be developed to  describe how a
forest algebra works on them.
A first approach consists, given  
a forest algebra $\calK = (K,U)$,
in regarding a
multicontext as a specification for a multivariate mapping
from $K^{B}$ to $K$. 
This defines the \emph{algebra of mappings}
$\calK_\#$; it is used to
define the notion of pumping, which underlies the construction
of certain Ehrenfeucht-Fra\"{\i}ss\'e games,
and to associate to \calK\
a threshold and a period that are consistent with those used
so far in the literature. 
A second approach consists in considering
that a port label specifies elements of $K$
are allowed as inputs at that port. 
This leads to the definition of the \emph{extended algebra} $\calK_\%$,
which we use to
generalize once more the notions of threshold, period,
and aperiodicity.
Necessary conditions for 
first-order definability, that supersede 
some of the existing ones,
are defined from the latter.
%

%
%

%

\subsection{Multicontexts}  \label{sec:redef}

We use both $m$ and  $(m,\nu)$ to denote the pair
consisting of a multicontext $m$, where 
every interior node carries a label $\lambda(y) \in A$,
and a port labeling $\nu(x) \in B$.
When this pair is equipped with a second port labeling
$\psi$, we denote the resulting tuple $t = (m,\nu,\psi)$ 
when $\psi$ is fixed and the emphasis is on $t$ as a whole,
and $\Brvpsi(m)$ when it is understood that 
$(m,\nu)$ is fixed and $\psi$ is one of several possible
second port labelings.
Next, instead of labeling a port directly with a horizontal
monoid element, as it was done in the previous section,
we take $\nu(x)$ and $\psi$ in sets $B$ and $C$,
respectively, such that
$A$, $B$ and $C$ are pairwise disjoint;
when dealing with specific algebras,
leaf extensions of the appropriate homomorphisms are
then defined on $B$ and $C$.
Note that we are ultimately interested in 
the recognition of languages over $A$,
so that $B$ and $C$ are artefacts used in this
process and the ultimate results should not depend on them.
The tuples  over  $A$ and $B\times C$,
along with the contexts defined from them
by replacing a leaf
with a special port $\square$,
constitute a forest algebra
$\BbbF_{A,B \times C} = (\BbbM_{A,B \times C},\BbbC_{A,B \times C})$;
those over  $A$ and $B$ constitute
$\BbbF_{A,B} = (\BbbM_{A,B} ,\BbbC_{A,B} )$;
the reader can verify that both
are free algebras. 
\newline
Besides the insertion in a context, i.e. the monoid operations in 
\Svert, $\BbbC_{A,B}$ and $\BbbC_{A,B \times C}$,
we define 
an operation that does multiple, simultaneous insertions in a multicontext from $\BbbM_{A,B}$.
%
%
Given sets of multicontexts $M_1$ and $M_2$ and $Z \subset ports(M_1)$,
with $Z \cap ports(m_1) \neq \emptyset$ for every $m_1 \in M_1$,
we denote by $ M_1 \, \underline{{\cdot}{\scriptstyle{Z}}}\,  M_2$
the set of all multicontexts $m_{\text{new}}$ that can be built by taking an element $m_1 \in M_1$,
inserting at each port $x \in Z \cap ports(m_1)$ a multicontext $m(x) \in M_2$
and replacing the label of $x$ with the neutral letter $e$;
with this new label, $x$ has no effect on the image by a homomorphism
of $m_{\text{new}}$, while it remains available to be used in reasonings
and proofs.
No other label is modified, so that in particular 
if $m(x)$ is a copy of $m_2 \in M_2$ and $y \in nodes(m_2)$, 
then the counterpart $y'$ of $y$ in  $m(x)$
satisfies $\lambda(m_{\text{new}},y') = \lambda(m(x),y) = \lambda(m_2,y)$.
\newline
Let $M \subset \BbbM_{A,B}$ and
let $Z \subseteq ports(M)$.
Then with $m \in M$,
we use the notations $Z(m)$ and $Z(M)$
for $Z \cap ports(m)$ and $Z \cap ports(M)$,
respectively. Given a congruence $\cong$,
we say that $Z$ is $\cong$-stable
when every pair $x,x'$ of ports satisfies
$\nabla(M,x) \cong \nabla(M,x') $.
\newline
Next, 
let $B' \subseteq B$ and let $Z $ be the set of all
ports with label in $B'$.
With $\theta \ge 1$ we define the set $M^{(\theta,Z)} $
obtained by pumping $\theta$ times the set $M$ at $Z$,by:
$M^{(1,Z)} = M$ and 
$M^{(\theta+1,Z)} = M^{(\theta,Z)} \, \underline{{\cdot}{\scriptstyle{Z}}}\,  M$.
This definition of pumping is consistent
with the definition of  the
vertical monoid of $\Sfree/{\cong}$
(where both $B$ and $Z$ are singletons),
 with
the
``vertical confusion'' defined in \cite{bostwa09}\
(where $B$ and $M$ are singletons),
and with  the 
``vertical confusion on uniform multicontexts'' 
also discussed in \cite{bostwa09}\
(where $B$ and $M$ are singletons and
the ports of $Z$ are indistinguishable
by any congruence).

%


 \subsection{The algebra of mappings}   \label{sec:algemap}

Let $\calK = (K,U)$ be a finite algebra and
$\gamma : \Sfree \rightarrow \calK$  a
surjective homomorphism.
We look for a reasonable way of extending \calK\ to 
$\BbbF_{A,B}$,
besides the one that consists in defining a leaf extension of $\gamma$ to $B$.
With this in mind, we define 
the \emph{algebra of mappings} of \calK, which we denote $\calK_\#$.
%
To do so, we
show how to translate a congruence in $\BbbF_{A,B} $
into a congruence in \Sfree, and vice versa.
Define a mapping $\hat{\varepsilon}_K$ from $\BbbM_{A,B}$ to $\Snat$,
that turns  $m \in \BbbM_{A,B}$ into a forest $\hat{\varepsilon}_K(m)$ by
replacing all port labels with the neutral letter $e$; define 
$\hat{\varepsilon}_V$ from $\BbbC_{A,B}$ to $\Svert$ in exactly the same way;
both mappings constitute a surjective homomorphism from  $\BbbF_{A,B}$ 
to \Sfree.
Thus, given a forest congruence $\approx$ over $\BbbF_{A,B}$, a forest congruence
is defined in a natural way over \Sfree.
%
%
In the other direction, let 
$B = \{b_1,\ldots,b_N\}$ be finite
and let $\xi $ be  simultaneously regarded as
a vector $\xi \in K^N$ a mapping $\xi : B \rightarrow K$.
Given $m \in \BbbM_{A,B}$,
define $\mu : ports(m) \rightarrow K$  by $\mu = \xi \circ \nu$.
From $m$ and  $\xi$, 
we  build a forest $\Brvmu(m,\xi)$ over $A \cup K$ by
replacing in $m$ every port label $\nu(x)$ with  $\xi(x)$.
We build $\Brvmu(w,\xi)$ from   $w \in \BbbC_{A,B}$ in the same way;
the $\square$-port of $w$   retains its label.
We extend $\gamma$ to $\BbbH_{A \cup K}$ by defining $\gamma(h) = h$
for every $h \in K$, 
so that $\xi$ in fact is one of the $|K|^N$ leaf extensions for $\gamma$ that can be 
built
on $B$,
and define a mapping $\gamma[m] :K^N \rightarrow K$
by
$\gamma[m](\xi) = \gamma(\Brvmu(m,\xi))$;
similarly, we define
$\gamma[w] :  K^{N} \rightarrow U$ by
 $\gamma[w](\xi) = \gamma(\Brvmu(w,\xi))$.
Next,
we define $K_\#$ and $U_\#$, the sets of all mappings
$ K^N \rightarrow K$ and $ K^N \rightarrow U$, respectively, and
given $f,f' \in K_\#$ and $u,u' \in U_\#$,
the operations and vertical action
$f+f' : \xi \mapsto f(\xi) + f'(\xi)$,
$uu' : \xi \mapsto u(\xi) u'(\xi)$, and
$uf: \xi \mapsto u(\xi) f(\xi)$,
so  that the pair $\calK_\#= (K_\# , U_\#)$
 constitutes a forest algebra\footnote{A
notation that mentions $B$, e.g.
$\calK_{\sharp B}$, would actually
be more accurate, if more cumbersome.
} 
  and  
$\gamma : \BbbF_{A,B} \rightarrow \calK_\#$
defined by $m \mapsto \gamma[m]$ and $w \mapsto \gamma[w]$
is a homomorphism.
Let  $\approx$ be  the nuclear congruence of $\gamma$:
we define from it an equivalence between 
nodes, also denoted  $\approx$.
Let $x,x' \in nodes(M)$ where $M$ is a set  of multicontexts
closed under $\approx$:
\begin{quote}
$x \approx x'$ iff $\lambda(x) = \lambda(x')$ and
$\Delta(M,x)   \approx \Delta(M,x')$ and
$\nabla(M,x)  \approx \nabla(M,x') $;
\end{quote}
nodes equivalent under this relation ``cannot be told apart by $\calK_\#$''.
We also write $(m,x) \approx (m',x')$ 
and $(M,x) \approx (M,x')$
in order to specify where the nodes are located.
%
%
Given $m \in M$
and $h_1,\ldots,h_{N-1} \in H$ we define 
the mapping
\[
\gamma[m; h_1,\ldots,h_{N-1}] : K \rightarrow K 
 \quad\ \text{by}\ \quad\
\gamma[m; h_1,\ldots,h_{N-1}](k) = \gamma[m]( h_1,\ldots,h_{N-1},k ) .
\]
Since $M$ is closed under $\approx$,
 $\alpha[m]$ is the same for every  $m \in M$
and we can use the notation $\gamma[M]$.
Then with $Z = \{ x \in ports(m) : \nu(x) = b_N \}$, we observe
\[
\gamma[M \, \underline{{\cdot}{\scriptstyle{Z}}}\,  M' ; h_1,\ldots,h_{N-1}]
\ = \ 
\gamma[M  ; h_1,\ldots,h_{N-1}]
\circ
\gamma[M'  ; h_1,\ldots,h_{N-1}] 
\]
\[
\text{and} \quad
\gamma[M \, \underline{{\cdot}{\scriptstyle{Z}}}\,  M' ] ( h_1,\ldots, h_N )
\ =\
\gamma [ M] ( h_1,\ldots,h_{N-1}, \gamma [ M'] ( h_1,\ldots, h_N )) .
\]
Therefore, every operation
$ \underline{{\cdot}{\scriptstyle{Z}}} $
satisfies the compatibility property
for $\approx$
\cite[Definition 5.1]{busa00}.

\begin{proposition} \label{prp:Hsharp}
Let $\calH^n_\#$ and $\alpha^n_\#$
be built from $\calH^n_{\tau,\pi} = \Sfree / \approx^n_{\tau,\pi} $ and $\alpha^n_{\tau,\pi}$.
Then the nuclear congruence of  $\alpha^n_\#$
is refined by the congruence $\approx^n_{\tau,\pi}$ built recursively
over $\BbbF_{A,B}$, 
as in Section \ref{sec:block}.
Hence, $\calH^n_\# \prec \BbbF_{A,B} / \approx^n_{\tau,\pi} $.
\end{proposition}

\begin{proof}
By induction on $n$.
Recall that a given $\xi$ is regarded
both as a mapping $\xi : B \rightarrow H^n_{\tau,\pi}$ and 
as a vector $\xi \in  (H^n_{\tau,\pi})^{|B|}$.
For the $n=1$ case, we associate to every $m \in \BbbM_{A,B}$
a vector $p_1(m)$ with component in $\N_{\tau,\pi}$ and labels in $A \cup B$,
where $p_1(m)_a$, $a \in A$, and $p_1(m)_b$, $b \in B$, 
are respectively the number of nodes $y$
with label $\lambda(y) = a$ and the number of ports $x$
with $\nu(x) = b$. The algebra
$H^1_{\tau,\pi} = \Sfree / \approx^1_{\tau,\pi} $  
is isomorphic to $(\N_{\tau,\pi} )^{|A|}$,
so that, with some abuse of notations, we can write
$\alpha^1_{\tau,\pi} (m) = p(m)$,
and given a mapping $\xi : B \rightarrow H^1_{\tau,\pi}$,
the image  $\xi(b)$ of $b \in B$
can be represented as a vector $\xi_1(b) \in (\N_{\tau,\pi} )^{|A|}$.
Within $ \BbbM_{A,B}$, there is an equivalence class
under  $\approx^1_{\tau,\pi} $  for every possible value of
$p_1(m)$, i.e. every vector in $(\N_{\tau,\pi} )^{|A \cup B|}$.
Then given $m \in \BbbM_{A,B}$,
\[
\alpha^1_{\tau,\pi} [m] (\xi) \ =\ 
\alpha^1_{\tau,\pi} (m) + \sum_{x \in ports(m)}\!\! \xi(\nu(x))
\ =\ 
\sum_{a \in A} p_1 (m)_a + \sum_{b \in B}  \xi_1(b) {\cdot} p_1(m)_b .
\]
From there, if
$p_1(m) = p_1(m')$, then $\alpha^1_{\tau,\pi} [m] = \alpha^1_{\tau,\pi} [m'] $.
%
%
With $n \ge 2$, 
the induction hypothesis states that if
$ m \approx^{n-1}_{\tau,\pi} m'$, then for every 
mapping $\xi : B \rightarrow H^{n-1}_{\tau,\pi} $,
the leaf completions of $m$ and $m'$ through $\xi$
satisfy
$ \Brvxi(m) \approx^{n-1}_{\tau,\pi}  \Brvxi(m')$.
Assume that 
$ m \approx^n_{\tau,\pi} m'$.
Two nodes or ports $x$ of $m$ and $x'$ of $m'$ receive the same label
in the  versions   of   $m$ and $m'$ relabeled according to
$\alpha^{n-1}_{\tau,\pi}$ iff
$\lambda(m,x) = \lambda(m',x')$,
$ \Delta(m,x) \approx^{n-1}_{\tau,\pi} \Delta(m',x') $ 
and
$  \nabla(m,x)  \approx^{n-1}_{\tau,\pi} \nabla(m',x')$.
By the induction hypothesis, the last two items imply, for every $\xi$:
\[
 \Brvxi(\Delta(m,x)) \approx^{n-1}_{\tau,\pi}  \Brvxi( \Delta(m',x') )
\quad \text{and} \quad
  \Brvxi(\nabla(m,x)) \approx^{n-1}_{\tau,\pi}  \Brvxi( \nabla(m',x') ) ,
 \]
which means that $x$ and $x'$  receive the same label
in the  versions   of 
$ \Brvxi(m) $ and $  \Brvxi(m')$ relabeled according to
$\alpha^{n-1}_{\tau,\pi}$, that is,
$ \Brvxi(m) \approx^{n}_{\tau,\pi}  \Brvxi(m')$.
 \vspace{0.2in}
\end{proof}

The algebras $\calH^n_\# $
and $ \BbbF_{A,B} / \approx^n_{\tau,\pi} $
are not isomorphic, however.
To see this, let $\tau =2$ and $\pi=1$,
so that  $\N_{2,1} = \{0,1,\infty\}$, 
 $n=1$ and
$A = \{a\}$,
so that 
$ \BbbF_{A} / \approx^1_{2,1} $
is isomorphic to $\ODalg_{\N_{2,1}}$,
where,
and finally
 $B = \{b\}$.
With $m = aab$ and $m' = aa(b+b)$,
we have $\alpha^1_{2,1}(m) \neq \alpha^1_{2,1}(m')$,
while
$\alpha^1_{2,1} [m] = \alpha^1_{2,1} [m'] $
is the constant function
that map $\{0,1,\infty\}$ onto $\infty$.

\subsection{Equivalence under pumping}   \label{sec:eqpump}

We use the algebra of mappings to  define a 
``threshold $\tau$, period $\pi$ equivalence under pumping''
congruence $\stackrel{\tau,\pi}{\leftrightarrow} = \bigcup_{i \ge 0} \stackrel{\tau,\pi}{\leftrightarrow}^{(i)}$
within $\BbbF_{A,B}$.
First, let $\stackrel{\infty}{\leftrightarrow}$ denote the relation where two forests are
equivalent iff they are the same up to horizontal permutations within a sum
(we might as well use $=$ instead of $\stackrel{\infty}{\leftrightarrow}$).
Then we consider a special case of
a multicontext  $m \in \BbbM_{A,B}$ where any two ports $x,x'$
satisfy $\nu(x) = \nu(x') \Leftrightarrow \nabla(m,x)  \stackrel{\infty}{\leftrightarrow} \nabla(m,x')$,
that is, their contexts within $m$ are indistinguishable.
Then the $\stackrel{\infty}{\leftrightarrow}$-stable sets of ports
are exactly the sets $\nu^{-1}(b)$, $b \in B$;
we say that $m$ is suitable for pumping.
Pumping\footnote{Note that this formalism also covers the case where
pumping is done ``horizontally'', i.e. where we 
are dealing with a multicontext of the form $m=c+x_1 + \cdots + x_k$
and where $Z=\{x_1 ,\ldots , x_k\}$. 
}
the singleton $\{m\}$ along 
a $\stackrel{\infty}{\leftrightarrow}$-stable
set of ports $Z$, we obtain for each
$\theta \in \N$  a singleton $\{m\}^{(\theta,Z)} = \{m^{(\theta,Z)} \}$.
We now define  $\stackrel{\tau,\pi}{\leftrightarrow}^{(i)}$, for
every $i \ge 0$.
First, 
$\stackrel{\tau,\pi}{\leftrightarrow}^{(0)}$ coincides with $\stackrel{\infty}{\leftrightarrow}$.
Next, the forest congruence $\stackrel{\tau,\pi}{\leftrightarrow}^{(1)}$ is 
generated by the pairs $(m^{(\theta,Z)},m^{(\theta',Z)})$ 
and the corresponding context congruence by the pairs
$(\nabla(m^{(\theta,Z)},x),\nabla(m^{(\theta',Z)},x'))  $,
where $\theta  \equiv_{\tau,\pi} \theta'$, $x \in Z(m^{(\theta,Z)})$ and $x' \in Z(m^{(\theta',Z)})$.
Then recursively for $i \ge 1$, given a set $M$ of multicontexts closed under
$\stackrel{\tau,\pi}{\leftrightarrow}^{(i)}$ and a set 
$Z \subset ports(M)$ that is $\stackrel{\tau,\pi}{\leftrightarrow}^{(i)}$-stable,
$\stackrel{\tau,\pi}{\leftrightarrow}^{(i+1)}$
is the congruence generated by the pairs
$(m,m')$ 
and
$(\nabla(m,x),\nabla(m',x'))  $
where $\theta  \equiv_{\tau,\pi} \theta'$, 
$m \in M^{(\theta,Z)}$, $m' \in M^{(\theta',Z)}$,
$x \in Z(m)$ and $x' \in Z(m')$.
We denote by $\calJ^{\tau,\pi} = ( J^{\tau,\pi} , U^{\tau,\pi})$ the 
(infinite)
quotient algebra
$\BbbF_{A,B} / {\stackrel{\tau,\pi}{\leftrightarrow}}$
and by
$\beta^{\tau,\pi}$ the corresponding surjective homomorphism.

%

%


\begin{proposition} \label{prp:thrper}
Every congruence of finite index over \Sfree\
is refined by a congruence 
$\stackrel{\tau,\pi}{\leftrightarrow}$.
\end{proposition}
 
\begin{proof}
Let $\calH = (H,V) = \Sfree / {\approx} = \alpha(\Sfree)$
be finite.
Let $B = \{b_1,\ldots,b_N\}$ and 
let  $m \in \BbbM_{A,B}$ be suitable for pumping.
Assume that
$Z = \nu^{-1}(b_N)$; we
pump the singleton $\{m\}$ along $Z$.
Given $h_1,\ldots,h_N \in H$ we define 
$g = \alpha[m]( h_1,\ldots,h_N )$ and the mapping
\[
\zeta : H \rightarrow H \ \quad\
\zeta(k) = \alpha[m]( h_1,\ldots,h_{N-1},k ) .
\]
Observe that $\alpha[m^{(2,Z)}]( h_1,\ldots,h_N ) = \zeta(g)$
and in general, 
\[
\alpha[m^{(\theta,Z)}]( h_1,\ldots,h_N ) = \zeta^{\theta-1}(g) .
\]
The mapping $\zeta$ generates a subsemigroup 
$\crochet{\zeta}$ of $H^H$;  
from the threshold and period of
$\crochet{\zeta}$ we obtain integers $\tau$ and $\pi$ 
such that, for all combination of $m$ and $Z$,
we have
$\alpha[m^{(\theta,Z)}] = \alpha[m^{(\theta',Z)}]$
as soon as
$\theta  \equiv_{\tau,\pi} \theta'$.
\newline
We prove by induction on $i$, that with these $\tau $ and $\pi$,
for every $\stackrel{\tau,\pi}{\leftrightarrow}^{(i)}$-closed set $M$
and every $\stackrel{\tau,\pi}{\leftrightarrow}^{(i)}$-stable $Z \subseteq ports(M)$,
every combination of multicontexts
$\hat{m} \in M^{(\theta,Z)}$ and
$\hat{m}' \in M^{(\theta',Z)}$
satisfies
$\alpha[\hat{m}] = \alpha[\hat{m}'] $, whenever
$\theta  \equiv_{\tau,\pi} \theta'$.
At $i=0$ 
the relation $\stackrel{\tau,\pi}{\leftrightarrow}^{(0)}$ coincides with $\stackrel{\infty}{\leftrightarrow}$.
If $M$ is $\stackrel{\infty}{\leftrightarrow}$-closed, then
every pair $m,m' \in M$ satisfies $\alpha[{m}] = \alpha[{m}'] $.
Now let $Z \subseteq ports(M)$ be  $\stackrel{\infty}{\leftrightarrow}$-stable.
For every pair $\hat{m},\hat{m}' \in M^{(\theta,Z)}$ and any $\theta \ge 2$,
we have $\alpha[\hat{m}] = \alpha[\hat{m}'] $,
in particular when $\hat{m} = m^{(\theta,Z)}$, which means
$ \alpha[\hat{m}'] = \alpha[m^{(\theta,Z)}]$ for all $\hat{m}' \in M^{(\theta,Z)}$.
Since
$ M^{(\theta,Z)} \cup  M^{(\theta',Z)}$ is $\stackrel{\tau,\pi}{\leftrightarrow}^{(1)}$-closed whenever
$\theta  \equiv_{\tau,\pi} \theta'$, every $\hat{m}' \in M^{(\theta,Z)} \cup  M^{(\theta',Z)}$  
satisfies simultaneously $\hat{m}' \stackrel{\tau,\pi}{\leftrightarrow}^{(1)} m^{(\theta,Z)}$
and $ \alpha[\hat{m}'] = \alpha[m^{(\theta,Z)}]$.
From this we can deduce that the proposition holds for $\stackrel{\tau,\pi}{\leftrightarrow}^{(1)}$.
The induction step $i \ge 1$ is proved in the same way.
 \vspace{0.2in}
 \end{proof}

Following \cite{bostwa09}, we say that multicontext $m$
is \emph{uniform} when $ports(m)$
is $\stackrel{\infty}{\leftrightarrow}$-stable, i.e
all nodes at every depth level have the same label and
the same number of sons. We denote by $\MVM\BbbV_{A}$
the set of all
uniform multicontexts over $A$; it can be seen as a subset of
$ \BbbM_{A,\{b\}}$ where $b \not\in A$. 
Define a symbol
$\underline{{\otimes}n}$
for each $n \in \N$, 
and let 
$\underline{{\otimes}\N} = \{ \underline{{\otimes}n} : n \in \N \}$;
regarding \Svert\ as another alphabet, we can see a 
uniform multicontext as the image of a word
of $\underline{{\otimes}\N} (\Svert \underline{{\otimes}\N} )^*$ by a mapping
$\zeta$ such that
$\zeta(v) = v$ for every $v \in \Svert$,
$\zeta( \underline{{\otimes}1} w) = \zeta(w) $,
$\zeta(\underline{{\otimes}2} w) = \zeta(w)  + \zeta(w) $, etc.
with notation $\zeta( \underline{{\otimes}n} w) = \underline{{\otimes}n} \zeta(w)$, and
$\zeta(ww') = \zeta(w) \, \underline{{\cdot}{\scriptstyle{Z}}}\, \zeta(w')$,
where $Z = ports(\zeta(w))$.
It is a standard exercise to verify that 
$\MVM \Svert$ is a monoid of transformations of \Snat\
and that
$\zeta$ is a 
homomorphism.
We then associate to $\calH = (H,V) = \alpha(\Sfree)$
its \emph{multivertical monoid},
defined as $\MVM V = \alpha(\MVM \Svert)$.
\newline 
The proof of Proposition \ref{prp:thrper}\
shows that
the smallest integers $\tau$ and $\pi$ such that
$\stackrel{\tau,\pi}{\leftrightarrow}$ refines $\approx$ 
are  the threshold and period of the 
multivertical monoid $\MVM V$.
Also, the algorithm described in
\cite{bostwa09}\ to decide whether an algebra has
vertical confusion on uniform multicontexts can be adapted to compute
the values of $\tau$ and $\pi$.
We generalize the condition tested by this algorithm
to ${{\mathbf{**}}}\bicrochet{\N_{\tau,\pi}}$
and the variety of monoids
$\mathbf{Sol}_{\tau,\pi}$,
generated by iterated block products
of elements of $\crochet{\N_{\tau,\pi}}$,
which means a generalization from the
aperiodic to the solvable monoids,

\begin{proposition} \label{prp:multivert}
If an algebra $\calG = (G,W)$ belongs to
${{\mathbf{**}}}\bicrochet{\N_{\tau,\pi}}$,
then its multivertical monoid belongs to
$\mathbf{Sol}_{\tau,\pi}$.
\end{proposition}

\begin{proof}
It suffices to prove that the
multivertical monoid of a block product of forest algebras
is a  block product  of their 
multivertical monoids.
%
%
%
Let $m \in \MVM \Svert$ and let $x \in nodes(m)$;
the subtree $s = \Delta^+(m,x)$ rooted at $x$ 
belongs to a sum of $1$ or more copies of 
$s$; let
$\mu(x)$ be the number of copies of $s$ in this sum,
other than  $\Delta^+(m,x)$ itself. Let
$\nabla^-(m,x)$ be the context in which this
sum is inserted, so that
$\nabla(m,x) = \nabla^-(m,x) ( \square + \underline{{\otimes} \mu(x) } s )$.
Let $\alpha : \Sfree \rightarrow \calH$
be a surjective homomorphism: the label of 
$x$ in the version $m^\alpha$ of $m$ relabeled
according to $\alpha$ is
\[
\lambda(m^\alpha,x) \ =\
\crochet{ \lambda(m,x), \
\alpha( \Delta(m,x) ) ,\
\alpha(  \nabla^-(m,x) ) ( \square +  \underline{{\otimes} \mu(x) }  \alpha(\Delta^+(m,x) ) )
} .
\]
Given $n \in \MVM \Svert$ and $y \in nodes(n)$, we have
\begin{eqnarray*}
\lambda( (m\, \underline{{\cdot}{\scriptstyle{Z}}}\, n )^\alpha ,x) = 
\crochet{ \lambda(m,x), \,
\alpha( \Delta(m,x)  \underline{{\cdot}{\scriptstyle{Z}}}\, n  ) ,\,
\alpha(  \nabla^-(m,x) \underline{{\cdot}{\scriptstyle{Z}}}\, n  ) 
( \square +  \underline{{\otimes} \mu(x) }  \alpha(\Delta^+(m,x) \underline{{\cdot}{\scriptstyle{Z}}}\, n  ) )
} \\
\lambda( (m\, \underline{{\cdot}{\scriptstyle{Z}}}\, n)^\alpha , y) = 
\crochet{ \lambda(n,y), \,
\alpha( \Delta(n,y) ) ,\,
\alpha( m\, \underline{{\cdot}{\scriptstyle{Z}}} \nabla^-(n,y) ) ( \square + 
 \underline{{\otimes} \mu(y) }  \alpha(\Delta^+(n,y) ) )
} ,
\end{eqnarray*}
and define ``above'' and ''below'' actions, such that
$\lambda( (m\, \underline{{\cdot}{\scriptstyle{Z}}}\, n )^\alpha ,x) = \lambda(m^\alpha,x) . \alpha(n)$
and
$\lambda( (m\, \underline{{\cdot}{\scriptstyle{Z}}}\, n)^\alpha , y) = \alpha(m) . 
\lambda(  n^\alpha , y) $.
The fact that the operation $ \underline{{\cdot}{\scriptstyle{Z}}}$ is associative
is used to verify
that these actions are monoidal; this ensures that
we indeed have a block product.
\end{proof}

 \subsection{The extended algebra}   \label{sec:extalg}

Let $\calK = (K,U)$ be a forest algebra and let
$\gamma : \Sfree \rightarrow \calK$ be surjective.
Given $(m,\nu) \in \BbbM_{A,B}$,
 we build all  forests $s$ that can be obtained by insertion
of elements of \Snat\
at the ports of $(m,\nu)$
and gather their images images $\gamma(s)$ in a set $\gamma_\%(m)$.
This is done under a restriction: we 
define a mapping $\gamma_\% : B \rightarrow \calP(K)$
and demand that the forest $t(x)$ inserted at $x$
satisfy $\gamma(t) \in \gamma_\%(\nu(x))$.
To simplify the reasonings and definitions, and
avoid having to deal with irrelevant special cases,
we assume that $\gamma_\%^{-1}(F) \neq \emptyset$
for every nonempty $F \subseteq K$, and
that there exists
a partition $C = \bigcup_{b \in B} C_b$
where $|C_b| = |\gamma_\%(b) |$ for every $b \in B$,
as well as a bijective leaf extension
$\gamma : C_b \rightarrow \gamma_\%(b)$.
We define
the notation $D = \bigcup_{b \in B} \{b\} \times C_b$,
so that we can write that we work on
$\BbbF_{A,D}$ rather than on $\BbbF_{A,B \times C}$,
and say that the $\gamma_\% $ defined above is the the
\emph{universal leaf extension} of $\gamma$
to $D$. 
We restrict our work to labelings 
$\psi : ports(m) \rightarrow C$ that are consistent with $\nu$,
in the sense that $\psi(x) \in C_{\nu(x)}$ for every port $x$.
Given $(m,\nu) \in \BbbM_{A,B}$, we denote by 
$\Psi(m,\nu) = \Psi(m)$  the set of all  mappings
$\psi$ consistent with $\nu$, and given
$\psi \in \Psi(m)$,
we use the notation $\Brvpsi(m,\nu)$ for the multicontext
$ (m,\nu,\psi) \in \BbbM_{A,D} $.
\newline
At some point we will look at more than one algebra at once,
e.g. \calG\ and \calK: then we will wlog assume that
the universal extensions $\varphi_\%$ and $\gamma_\%$
are defined in such a way that 
$(\varphi_\%^{-1}(F) \cap
\gamma_\%^{-1}(F) ) \neq \emptyset$
for every nonempty $F \subseteq G \times K$.
\newline
From the universal extension of $\gamma $ we define the mapping
$\gamma_\%   : \BbbM_{A,D} \rightarrow \calP(K)$
by
$\gamma_\%  (m,\nu) =  \{ \gamma ( \Brvpsi(m,\nu)) : \psi \in \Psi(m,\nu) \} $.
This is a homomorphism:
we verify this with the horizontal addition,
leaving the rest of the proof to the reader.
\begin{eqnarray*}
\gamma_\% (m) + \gamma_\% (m')
& = & 
 \{ \gamma ( \Brvpsi(m)) : \psi \in \Psi(m) \} 
 +
  \{ \gamma ( \Brvpsi'(m')) : \psi' \in \Psi(m') \}  \\
& = & 
 \{ \gamma ( \Brvpsi(m)) +   \gamma ( \Brvpsi'(m')) : \psi \in \Psi(m) ,\, \psi' \in \Psi(m') \} \\
& = & 
 \{ \gamma ( \Brvpsi(m) + \Brvpsi'(m')) : \psi \in \Psi(m) ,\, \psi' \in \Psi(m') \} \\
 & = & 
 \{ \gamma ( \Brvpsi''(m+m')) : \psi'' \in \Psi(m+m') \} \\ 
 & = & 
\gamma_\% (m+m')   .
\end{eqnarray*}
The image of  $\BbbF_{A,D}$ by $\gamma_\%$
is  the \emph{extended algebra} of \calK, 
with the notation $\calK_\%  = (K_\%,U_\%)$.
In the special case of ${\calH}{}^n_{\tau,\pi}$, 
since $ \tau,\pi$ are fixed parameters
most of the time,
we use the 
notations $\calH^n_\%$ and  $\alpha^n_\%$ 
instead of ${\calH}{}^n_{\tau,\pi}{}_\%$
and
${\alpha}{}^n_{\tau,\pi}{}_\%$.
%
%


\begin{proposition} \label{prp:KtoKbis}
Let \calH\ and
\calK\ be finite forest algebras.
If $\calK \prec \calH $,
then
$\calK_\%  \prec \calH_\% $.
\end{proposition}

\begin{proof}
Let $\gamma : \Sfree \rightarrow \calK $ and
$\alpha : \Sfree \rightarrow \calH$ be
be surjective homomorphisms.
Assume
that for all $t,t' \in \BbbM_{A, D }$,
$\alpha(t) = \alpha(t') $ implies
$\gamma(t) = \gamma(t')$.
Then with $m,m' \in \BbbM_{A, B  }$, we have
\begin{eqnarray*}
\alpha_\% (m) \subseteq \alpha_\% (m')
&  \Leftrightarrow & 
\forall \psi,\, \exists \psi' : 
\alpha \circ \Brvpsi(m) = \alpha \circ \Brvpsi'(m') \\
& \Rightarrow &
\forall \psi,\, \exists \psi' : 
\gamma \circ \Brvpsi(m) = \gamma \circ \Brvpsi'(m') \\
& \Rightarrow &
\gamma_\%  (m) \subseteq \gamma_\%  (m')   .
\end{eqnarray*}
Inclusion in the other direction is proved in the same way,
so that if
$\alpha_\% (m) = \alpha_\% (m')$,
then
$\gamma_\%  (m) = \gamma_\%  (m') $.
\end{proof}

\begin{proposition} \label{prp:HtoH}
For every $m,m' \in \BbbM_{A,B}$ and $n \ge 1$, if
$\alpha^n_\% (m ) = \alpha^n_\% (m' ) $,
then
$m \approx^n_{\tau,\pi}  m'$.
\end{proposition}

\begin{proof}
%
By induction on $n$.
For the $n=1$ case, 
as in the proofs of Propositions \ref{prp:Hsharp},
we associate to every $m \in \BbbM_{A,B}$
a vector $p_1(m)$ with component labels in $A \cup B$,
where $p_1(m)_a$, $a \in A$, and $p_1(m)_b$, $b \in B$, 
are respectively the number of nodes $y$
with label $\lambda(y) = a$ and the number of ports $x$
with $\nu(x) = b$; we define for $t \in \BbbM_{A,D}$ a vector 
$p_1(t)$ in the same way, with component labels in
$A \cup C$.
Let $m \in \BbbM_{A,B}$: its image 
$\alpha^1_\%(m) = \{ {\alpha}^{1}_{\tau,\pi} (\Brvpsi(m)) : \psi \in \Psi(m) \}$
is determined by the set of all vectors
$p_1(\Brvpsi(m))$ that can be obtained from $p_1(m)$.
While $p_1(m)_a = p_1(\Brvpsi(m))_a$ for every $a \in A$,
given $b \in B$ we have 
\[
p_1(m)_b = \sum_{c \in C_b} p_1(\Brvpsi(m))_c ,
\]
so that for any    $\psi' \in \Psi(m')$,
$ p_1(\Brvpsi(m))  \equiv_{\tau,\pi} p_1(\Brvpsi'(m'))$ 
implies
$p_1(m) \equiv_{\tau,\pi} p_1(m') $,
from which we conclude 
$\alpha^1_\%(m) =  \alpha^1_\%(m')  \Rightarrow m \approx^1_{\tau,\pi} m'$.
\newline
Let $n \ge 2$:
the induction hypothesis
states that if
${\alpha}^{n-1}_{\tau,\pi} (\Brvpsi(m) )= 
{\alpha}^{n-1}_{\tau,\pi} (\Brvpsi'( m')) $
for some $\psi \in \Psi(m)$ and $\psi' \in \Psi(m')$,
then
$m \approx^{n-1}_{\tau,\pi} m'$.
Let 
$p_n(\Brvpsi(m))$ denote the 
vectors  with a component for every
element of the appropriate relabeling alphabet,
which counts  the number of its occurrences 
in the  version of   $\Brvpsi(m)$ relabeled according to
$\alpha^{n-1}_{\tau,\pi}$.
%
Then
\begin{eqnarray*}
\alpha^n_\% (m) = \alpha^n_\% (m') 
& \Rightarrow  & 
\forall\, \psi,\, \exists\, \psi' : 
\alpha^n_{\tau,\pi} ( \Brvpsi (m )) = \alpha^n_{\tau,\pi} ( \Brvpsi' (m' )) \\
& \Rightarrow  & 
\forall\, \psi,\, \exists\, \psi' : 
p_n ( \Brvpsi (m) )   \equiv_{\tau,\pi}  p_n ( \Brvpsi' (m' )) .
\end{eqnarray*}
Two nodes $x$ and $x'$ 
in $ \Brvpsi (m) $  and $ \Brvpsi' (m') $
carry  the same 
element of the relabeling alphabet
only if they carry the same label in $m$ and $m'$,
and
\[
{\alpha}^{n-1}_{\tau,\pi} ( \Delta(\Brvpsi (m),x)) = 
{\alpha}^{n-1}_{\tau,\pi} ( \Delta(\Brvpsi' (m'),x')) 
\quad \text{and} \quad
{\alpha}^{n-1}_{\tau,\pi} (  \nabla(\Brvpsi (m),x)) = 
{\alpha}^{n-1}_{\tau,\pi} (  \nabla(\Brvpsi' (m'),x')) ;
\]
by the induction hypothesis, this implies
\[
\Delta(m,x)) \approx^{n}_{\tau,\pi} \Delta(m',x') 
\quad \text{and} \quad
 \nabla(m,x)) \approx^{n}_{\tau,\pi}  \nabla(m',x') ,
\]
hence $x$ and $x'$
carry with the same 
label in the
relabeled versions of $m$ and $m'$.
Therefore,
$\alpha^n_\% (m) = \alpha^n_\% (m')
\Rightarrow
m \approx^{n}_{\tau,\pi} m'$.
%
%
%
\end{proof}

\begin{proposition} \label{prp:HbarFO}
For every $n \ge 1$,
 $\calH^n_\% \in  {{\mathbf{**}}}\crochet{\N_{\tau,\pi} }$.
\end{proposition}

\begin{proof}
By induction on the structure of the horizontal monoid of
$\calH^n_\%$; the induction base
is done on  $\calH^1_\%$, however.
\newline
As in the proofs of Propositions \ref{prp:Hsharp}\ and \ref{prp:HtoH}, we
associate to every $m \in \BbbM_{A,B}$
a vector $p_1(m)$,
where it was seen that
for every $b \in B$ we have 
\[
p_1(m)_b = \sum_{c \in C_b} p_1(\Brvpsi(m))_c .
\]
Then 
$ \alpha^1_\%(m) =  \alpha^1_\%(m') $
when $p_1(m)_a  \equiv_{\tau,\pi} p_1(m')_a $ for every $a \in A$
and
$p_1(m)_b  \equiv_{\sigma,\pi} p_1(m')_b $
for every $b \in B$, where $\sigma \ge \tau {\cdot} |K|$, so that 
having at least $\sigma$ ports with 
$\nu(x) = b$ ensures that
there is enough room for $\tau$ occurrences of $\psi(x) = c$,
for every $c \in C_b$.
Therefore,  $\calH^1_\% \in  {{\mathbf{**}}^1}\bicrochet{\N_{\sigma,\pi} }$.
\newline
At the induction step, let $I$ be the minimal ideal of the horizontal monoid
$H^n_{\tau,\pi}$. We assume the
existence of $q \ge n$ such that for every $h \not\in I$,
the language 
$L_h = \{ m \in \BbbM_{A,B}  : h \in  \alpha^n_\%(m) \}$
is recognized by an algebra  $\calK \in {{\mathbf{**}}^q}\bicrochet{\N_{\sigma,\pi} }$.
%
%
Let $m  \in \BbbM_{A,B} $ for which we know that 
$ \alpha^n_\% (m) \cap I \neq \emptyset$,
that is,
$ {\alpha}^n_{\tau,\pi}(\Brvpsi(m)) \in I$ for at least one $\psi \in \Psi(m)$.
We say that a node $y $ of $m$ is a \emph{pathhead} 
for $\calH^n_\%$
when
$ \alpha^n_\% (\Delta^+(m,y)) \cap I \neq \emptyset$
and
$  \alpha^n_\%(\Delta^+(m,z)) \subseteq  H^n_{\tau,\pi} \setminus I$
for every son $z$ of $y$.
We  build from $m$ 
a multicontext $\overline{m}$ whose
interior nodes are the pathheads and their ancestors,
and where a port $x$ is created 
along every edge $(y,z)$ such that
$ y$ is an interior node of $\overline{m}$ and 
its son $z$ is not, so that 
$ \Delta(m,x)  = \Delta^+(m,z) $.
Given $k \in I$
and a mapping $\chi : ports(\overline{m}) \rightarrow H^n_{\tau,\pi} \setminus I$,
the algebra $\calH^n_{\tau,\pi}$ can recognize  
whether the resulting forest $\Brvchi(\overline{m})$ 
belongs to the set $ ({\alpha}^n_{\tau,\pi})^{-1}( k)$.
Next, for every pair $z,z' \in ports(\overline{m}) $ and  $h,h' \not\in I$, 
whether $h \in {\alpha}^n_\%(\Delta(m,z) )$
is determined by  \calK, and the same holds for the question of
whether 
$h' \in {\alpha}^n_\%(\Delta(m,z') )$.
The two hold simultaneously iff there exist 
labelings $\psi \in \Psi(\Delta(m,z) )$ and $\psi' \in \Psi(\Delta(m,z') )$  
that 
ensure
$\alpha^n_{\tau,\pi} \circ \Brvpsi(\Delta(m,z)) = h$ and
$\alpha^n_{\tau,\pi} \circ \Brvpsi'(\Delta(m,z')) = h'$.
Since the sets
$ports(\Delta(m,z))$ and $ports(\Delta(m,z'))$ are disjoint,
the hypotheses on $\psi$ and $\psi'$
are independent, and can be reworded as the 
existence of a suitable labeling in
$\Psi(\Delta(m,z)+\Delta(m,z'))$.
\newline
From the induction hypothesis, 
whether a node  is a pathhead is recognized by an algebra
of the form $\calK \Blockprod \ODalg_M$, with
$M \in \crochet{\N_{\sigma,\pi} }$;
a further block product 
$(\calH^n_{\tau,\pi}  \times (\calK \Blockprod \ODalg_M) )\Blockprod \ODalg_{M'}$, 
with
$M' \in \crochet{\N_{\sigma,\pi} }$,
recognizes whether $k \in \alpha^n_\%(m)$,
where $k \in I$.
%
\end{proof}

\begin{proposition}   \label{prp:firstSD}
If $\calG \in  {{\mathbf{**}}}\bicrochet{\N_{\tau,\pi} }$,
then
 $\calG_\% \in  {{\mathbf{**}}}\bicrochet{\N_{\tau,\pi} }$
 and
 $\MVM W_\% \in \mathbf{Sol}_{\tau,\pi}$.
\end{proposition}

\begin{proof}
This is a consequence of Propositions
\ref{prp:multivert},
\ref{prp:KtoKbis}, and \ref{prp:HbarFO}.
\vspace{0.02in}
\end{proof}

Let $\sigma$ and  $\rho$ be  the smallest
integers  such that
$  \stackrel{\sigma,\rho}{\leftrightarrow} $
refines
 the canonical congruence of
$\calG_\%  $:
given Proposition \ref{prp:firstSD},
it makes sense to call them the threshold and period of $\calG_\%  $.
%
%
%
In the special case of $ {{\mathbf{**}}}\bicrochet{\N_{\tau,1} }$,
the condition  $\MVM W_\% \in \mathbf{A}$
supersedes both the aperiodicity of $W$
and the absence of vertical confusion on uniform multicontexts.
As mentioned earlier, Potthoff's algebra, 
described in Section \ref{sec:potthoff},
satisfies these two conditions while
its extended algebra has a   non-aperiodic vertical monoid.



\section{Recursive Proofs} \label{sec:main}

In this section,  we consistently use the following  notations.
The
algebra $\calG = (G,W)$
is the one whose membership the variety ${{\mathbf{**}}}\bicrochet{\N_{\tau,\pi}}$ is to be decided.
We associate to \calG\ a surjective homomorphism 
$\varphi : \Sfree \rightarrow \calG$ with nuclear congruence $\simeq_G$.
Given  $n \ge 1 $
and the variety ${{\mathbf{**}}}^n\bicrochet{\N_{\tau,\pi}}$,
let $\approx^n_{\tau,\pi}$, 
$\calH^n_{\tau,\pi} = \Sfree / {\approx^n_{\tau,\pi}}$
and $\alpha^n_{\tau,\pi} : \Sfree \rightarrow \calH^n_{\tau,\pi}$,
such that  $\approx^n_{\tau,\pi}$ is
the finest  congruence over \Sfree\  that satisfies
$\calH^n_{\tau,\pi} \in {{\mathbf{**}}}^n\bicrochet{\N_{\tau,\pi}}$
and
$\alpha^n_{\tau,\pi}$ is a surjective homomorphism. 
We use the notation ${{\mathbf{**}}}^n\bicrochet{\N_{\tau,\pi}}$
both for the variety of algebras and for the variety of
the forest or context languages that they recognize.

\subsection{$\approx_{\tau,\pi}^\N$-tight sets}   \label{sec:prelim}

It is already folklore  that non-membership in 
$\mathsf{FO}[\Panc] $
of 
an algebra $\calG = (G,W)$ ultimately has to do with mutually accessible elements of its
horizontal monoid, and to \calR-classes of $W$.
The following preliminaries confirm this, and introduce some
definitions and facts that are used later.
\newline
For every  $g \in G$ we define the language $L_g = \varphi^{-1}(g) $ 
and, for $K \subset G$,  $L_K = \bigcup_{k \in K} L_k$.
It is a basic fact that for any alphabet  $A$ of size at least $2$ and variety of algebras $\bfW$,
 $\calG \not\in \bfW$ iff for every surjective homomorphism $\varphi : \Sfree \rightarrow \calG$,
there exists an element $g \in G$ such that $L_g \not\in \sfW$.
Therefore, $\calG \not\in\ {{\mathbf{**}}}\bicrochet{\N_{\tau,\pi}}$ implies the existence of a partition of $G$ into 
$
\Gout = \{ g \in G : L_g \not\in\ {{\mathbf{**}}}\bicrochet{\N_{\tau,\pi}} \}
$
and its complement \Gin. 
Since we always work with $\tau \ge 1$, 
both 
$L_0 = \varphi^{-1}(0) $
and the  set \Sstar\ of all trees over $A$
belong to $ {{\mathbf{**}}}\bicrochet{\N_{\tau,\pi}}$,
so that  $\Gin \neq \emptyset$.
Given   $g \in \Gout$,
we define
$\Lmin_g $ as the subset of $L_g \cap \Sstar $
where $\varphi(\Delta^+(y)) \not\in \Gout$
for every non-root node.
A \emph{pathhead} for $g$
is a node $x $ such that
$\Delta^+(s,x) \in \Lmin_g$; there exists
an algebra in ${{\mathbf{**}}}\bicrochet{\N_{\tau,\pi}}$ that
can recognize  whether a node of $s$ is a
pathhead for $g$. We say that $g \in \Gout$
is minimal
when at least one forest in \Snat\ contains a pathhead for $g$. 
A subset of \Gout\ is minimal if 
it contains at least one minimal element.
There exists therefore  
an algebra in ${{\mathbf{**}}}\bicrochet{\N_{\tau,\pi}}$  which recognizes every language $L_h$,
as well as every $\Ltree_h = L_h \cap \Sstar$, $h \in \Gin$
and every $\Lmin_g$ for $g \in \Gout$ minimal.
%
%
We now show that in
every $\calG \not\in {{\mathbf{**}}}\bicrochet{\N_{\tau,\pi}}$ the set \Gout\
has at least one minimal
strongly connected  subset.

\begin{proposition} \label{prp:cfc}
In  $\calG  \not\in {{\mathbf{**}}}\bicrochet{\N_{\tau,\pi}}$,
at least two elements of $\Gout$ are mutually accessible.
\end{proposition}

\begin{proof}
Assume the existence of $g \in \Gout$ such that $g \le h$ does  not hold for
any $h  \in \Gout$, $h \neq g$. This implies that $g \le g+h$ does not hold either,
so that in particular, either $g+g=g$ or $g+g < g$.
We assume that $g+g=g$; the other case is dealt with in a similar, and simpler, manner.
Consider a forest $t = s_1 + \cdots + s_n$ with each $s_i \in \Sstar$, and let $t \in L_g$;
there are two possible cases for $t$. First,  none of the $s_i$'s belongs to $\Ltree_g$
and whether $t \in L_g$ is determined by
counting, for each $h \in \Gin$, the number of trees in $t$ that belong to
$\Ltree_h$; this is done by a block product $\calHin \Blockprod \ODalg_M$
with $M \in \crochet{\N_{\tau,\pi}}$.
In the second case, at least one tree in $t$ satisfies $s_i \in \Ltree_g$.
This means satisfying two conditions: that
$s_i \in L_g^{\mathsf{pathhead}} - \bigcup_{k \neq g} L_k^{\mathsf{pathhead}} $,
where 
$L_g^{\mathsf{pathhead}}$ is the set of all trees that
contain a pathhead for $g$;
that
the context within $s_i$ of every pathhead maps $g$ to itself,
which makes it belong to  
 the subset of \Svert\ generated  by 
 \[
\{\ a\square\ |\ (\varphi(a\square))g=g\ \} \cup 
\{\ \square + t\ |\ t \in L_g^{\mathsf{pathhead}}\ \}
\cup \bigcup_{h \in G_X} \{\ \square + t\ |\ t \in L_h\ \},
\]
where 
 $G_X = \{h \in \Gin \;|\; g+h = g \}$.
Each set mentioned here is recognizable by an algebra in ${{\mathbf{**}}}\bicrochet{\N_{\tau,\pi}}$.
Finally,  a forest with at least one tree in $\Ltree_g$ and no tree  in
$\Ltree_k$ for any $k \in \Gout$   will belong to $L_g$  iff everyone of its other trees
belongs to $\bigcup_{h \in G_X} \Ltree_h$.
\vspace{0.05in}
\end{proof}
%


%
%
%
Given $J \subset G$, 
we say that a set of forests  $S$ is 
\emph{diagonal} for $J$ if $\varphi$ works as a bijection 
 from $S$ to $J$, i.e.
we can write
$S = \{ s_j : j \in J \}$ with $\varphi(s_j) = j$ for every $j \in J$.
If
$ \calG \not\in {{\mathbf{**}}}\bicrochet{\N_{\tau,\pi}}$, then
there exists 
for every   $n \ge 1$ a non-singleton 
$J^{(n)} \subseteq G$ and a set of forests
$S^{(n)} = \{ s_{j}^{(n)} : j \in J^{(n)} \}$, 
closed under  $\approx_{\tau,\pi}^n $-closed and
diagonal for $ J^{(n)}$,
which can serve as witnesses for $\approx_{\tau,\pi}^n $.
At least one $J \subseteq G$
occurs infinitely often in the sequence, and since
$\approx_{\tau,\pi}^n $ refines every $ \approx_{\tau,\pi}^m $, $m<n$,
a sequence of witnesses exists where $J^{(n)} = J$ for every $n$.
A strongly connected $J \subseteq G$ with this property is said
to be 
$\approx_{\tau,\pi}^\N$-tight.
We denote by \calJ\ the set of all $\approx_{\tau,\pi}^\N$-tight
subsets of $G$.


\begin{proposition} \label{prp:condA}
An algebra $\calG =(G,W)$ is outside of ${{\mathbf{**}}}\bicrochet{\N_{\tau,\pi}}$ iff
$G$ has  
a  nonsingleton $\approx_{\tau,\pi}^\N$-tight subset.
\end{proposition}

\begin{proof}
It suffices to prove that if  $\calG \not\in {{\mathbf{**}}}\bicrochet{\N_{\tau,\pi}}$,
then  $G$ contains a $\approx_{\tau,\pi}^\N$-tight pair.
Assume that $\calG \not\in {{\mathbf{**}}}\bicrochet{\N_{\tau,\pi}}$ and that
in every
set of witnesses where $\varphi(s_i)=k$ and $\varphi(s'_i) = k'$ for every $i$,
at least one of $k$ and $k'$ 
is not accessible from the other. 
Let this be $k$: denote by 
$F$  the  strongly connected component  to
which $k$ belongs, so that
$k$ is only accessible from elements of $F \cup \Gin$, 
$k' \not\in F \cup \Gin$
and $F \neq \Gout$. 
We denote by $\calG'$ the image of \calG\ by the homomorphism
that maps $G - (F \cup \Gin)$ onto an absorbing element $\infty$
and works as the identity on $F \cup \Gin$.
Then $\calG' \not\in {{\mathbf{**}}}\bicrochet{\N_{\tau,\pi}}$, $\Gout' = F \cup \{\infty\}$ and 
the hypothesis implies that
$\infty$ belongs to every $\approx_{\tau,\pi}^\N$-tight pair of $\calG'$.
Since no subset $\{h,k\}$ of $F$ is $\approx_{\tau,\pi}^\N$-tight, 
there exists a large enough $n$
such that $\calH^n_{\tau,\pi} = \alpha^n_{\tau,\pi}(\Sfree)$
recognizes $\Gin$ and, 
given a forest $s \in \varphi^{-1}(\Gout')$,  
determines  which subset $L_g \cup L_\infty$ it belongs to.
This means that for each $g \in G'$, there exists a first-order formula $\Phi_g(s)$ that takes
value true iff $s \in \Snat$ satisfies 
$s \in L_g $ when $g \in \Gin$, and
$s \in L_g \cup L_\infty$ when $g \in F$.
\newline
Let $s \in L_\infty$ have a  subtree $t \in L_\infty$ 
which is minimal in the sense that no strict subforest of $t$ is in $L_\infty$;
this $t$ is of the form $x {\cdot} t'$ where $x$ is the root of $t$,
$\lambda(x) = a$
and $g = \varphi(t') $,
with 
$(\varphi(a\square))g = \infty$. 
If  $g \in \Gin$, then 
the formula $\Phi_g(t') \wedge (\lambda(x) = a)$
asserts that
$\varphi(s) = \infty$.
Otherwise, $g \in F$: this holds only if for every $y \in nodes(t')$
such that $\varphi(\Delta^+(t',y)) \in \Gout'$,
there exist $k \in F$ and $b \in A$
such that $\Phi_k(\Delta(t',y))$ is true,
$\lambda(y) = b$ and $(\varphi(b\square))k \neq \infty$.
Therefore, the existence in $s$ of a minimal subtree
$t \in L_\infty$ can be asserted with a first-order formula.
Next, the formulas developed in this way are used to deal
with the case where
$s \in L_\infty$ has no subtree $t \in L_\infty$, i.e.
$s = t_1 + \cdots + t_p$ with  $\varphi(t_i) \in \Gin \cup F$
for each $i$. 
\end{proof}
%
%

%
\begin{proposition} \label{prp:Nzero}
Let algebra $\calG =(G,W)$ be outside of ${{\mathbf{**}}}\bicrochet{\N_{\tau,\pi}}$.
There exists an integer $n_0$ such that, for  every 
strongly connected set $J \subseteq \Gout$,
if  there exists  a set  of witnesses for $J$ and ${\approx^{n_0}_{\tau,\pi}}$,
then $J$  is $\approx_{\tau,\pi}^\N$-tight.
\end{proposition}

\begin{proof}
Let $\calG \not\in {{\mathbf{**}}}\bicrochet{\N_{\tau,\pi}}$ and let 
$J  \subseteq \Gout$ be strongly connected.
If $J$ is not $\approx_{\tau,\pi}^\N$-tight, then 
no set  of witnesses exists for $J$ and ${\approx^{n}_{\tau,\pi}}$,
for some $n \in \N$, and no such set exists either for any $n' \ge n$. 
Denote by $n_J$ the smallest integer with this property.
Since $G$ is finite, there is a maximum $N$ for the value of $n_J$ over all 
subsets $J$ that are not $\approx_{\tau,\pi}^\N$-tight.
Define $n_0 = N+1$:  the existence of a set of witnesses
for $J $  and ${\approx^{n_0}_{\tau,\pi}}$
implies that $J$ is $\approx_{\tau,\pi}^\N$-tight.
\end{proof}
%

%
%

From now on, we denote by $n_0$ an integer that satisfies the
conditions of Proposition \ref{prp:Nzero}, and such that
${\calH^{n_0}_{\tau,\pi}}$  
recognizes every language $L_h$,
as well as every $\Ltree_h = L_h \cap \Sstar$, $h \in \Gin$
and every $\Lmin_g$ for $g \in \Gout$ minimal.


%


\subsection{Recursive proofs for non-membership}     \label{sec:RP}

In a recursive proof, a set $\calS^{(n+1)}$ of witnesses is built by inserting
copies of elements of $\calS^{(n)}$ in the components of a circuit
$\calM^{(n+1)}$.
This circuit component is a tuple, i.e. a multicontext equipped with
suitable port labelings, whose purpose is to specify, for
every port $x$, which witness from  $\calS^{(n)}$, or rather copy thereof, is to
be inserted at $x$. The
tuples in a circuit must  be related to each other,
in such a way that certain of the resulting forests 
mapped by $\varphi$  to distinct elements of the same 
$\approx_{\tau,\pi}^\N$-tight subset of $G$
are undistinguishable by the congruence
$\approx^{n+1}_{\tau,\pi}$.
From this constraint we define a relation
$ \crochet{ \star\, \calH^n_{\tau,\pi}  }$
 between tuples, such that if
 $t \crochet{ \star\, \calH^n_{\tau,\pi}  } t'$, then the forests
built from $t$ and $t'$ are equivalent under $\approx^{n+1}_{\tau,\pi}$.
This leads to the definition of Condition
$\mathbf{RC}(\calG \star \calH^n _{\tau,\pi})$, which states
the existence of a circuit suitable
to the construction of witnesses for the congruence
$\approx^{n+1}_{\tau,\pi}$,
so that $\mathbf{RC}(\calG \star \calH^n _{\tau,\pi})$
implies 
$\calG \not\in {{\mathbf{**}}^{n+1}}\bicrochet{\N_{\tau,\pi}}$.
\newline
In the other direction, Lemma \ref{lem:secondRC}\ shows that  a circuit that
satisfies $\mathbf{RC}(\calG \star \calH^n _{\tau,\pi})$
can always be extracted from a set of witnesses for $\approx^{n+1}_{\tau,\pi}$.
As a consequence, a recursive proof of
non-membership exists for every algebra outside of
${{\mathbf{**}}}\bicrochet{\N_{\tau,\pi}}$. 
To prove
Lemma \ref{lem:secondRC}\ it is convenient to work in terms of ``full'' proofs,
where the number of witnesses in $\calS^{(n)}$ increases with $n$;
this does not describe the
actual proofs that exist in the literature. 
Theorem \ref{thm:RC}\ shows that a  ``slender''
proof always exists for every algebra outside of
${{\mathbf{**}}}\bicrochet{\N_{\tau,\pi}}$. 
\\

We define  a mapping
$\iota : \bigcup_{n \ge n_0} H^{n}_{\tau,\pi} \rightarrow \calP(G)$, given by
$\iota(h) =  \{ j \in G : h \in \alpha^{n}_{\tau,\pi} (L_j ) \}$,
which
is not necessarily injective.
By our assumption on $n_0$, either
$\iota(h) \subseteq \Gin$ and  is a singleton,
or
$\iota(h) \subseteq \Gout$ 
and may have  two or more elements.
If $|\iota(h)| \ge 2$, then every forest in
the set  $(\alpha^{n}_{\tau,\pi})^{-1}(h)$ can be used as a witness
for  $\iota(h)$ and
$\approx^{n}_{\tau,\pi} $, and therefore
by Proposition \ref{prp:Nzero}, $\iota(h)$
is $\approx_{\tau,\pi}^\N$-tight.
We define
$\calJ = \{\, \iota(h) : |\iota(h)| \ge 2\, \}$ and,
for each $n \ge 2$ the set
$\calD^{(n)} = \{\, (h,j) : h \in H^{n}_{\tau,\pi} \wedge\, j \in \iota(h) \, \}$.
We work in
terms of $\calD^{(n)}$ instead of 
$ \{\, (h,j) \in \calD^{(n)} : \iota(h) \in \calJ \}$,
that is, we 
include in the discussion those $h \in H^{n}_{\tau,\pi}$
for which $\iota(h)$ is a singleton;
this will be useful in the proof of Lemma \ref{lem:secondRC}.
%
We say that a set of forests $\calS^{(n)}$ is \emph{full}
when it contains a forest $s_{k,\ell}$ for every 
$(k,\ell) \in \calD^{(n)}$,
with the notations
$\calS^{(n)} = \{ \, s^{(n)}_{k,\ell} : (k,\ell) \in \calD^{(n)} \, \}$. 
Similarly, we say that a circuit over 
$(A,\calD^{(n)}, \calD^{(n+1)})$
is full when it contains a multicontext
$m^{(n+1)}_{h,j}$ for every 
pair $(h,j) \in \calD^{(n+1)}$.
%
\newline
A sequence $\calS^{(n)}$, $n \ge 1$, is 
a recursive proof when every forest $s^{(n+1)}_{h,j}$,
$h \in H^{n+1}_{\tau,\pi}$, $j \in \iota(h)$
is built by inserting copies of elements of
$\calS^{(n)}$ at the ports of a 
multicontext $m^{(n+1)}_{h,j}$ that belongs to a circuit $\calM^{(n+1)}$
over $(A,\calD^{(n)},\calD^{(n+1)})$.
%
Let $s$ and $m$ be shorthands for   $s^{(n+1)}_{h,j}$
and $m^{(n+1)}_{h,j}$, respectively.
To every $x \in ports(m)$, we assign  two labels 
$\mu^{n}_{\tau,\pi} (x) \in H^{n}_{\tau,\pi} $
and
$\psi(x) \in G$.
The labels are \emph{consistent}
when 
$\psi(x) \in  \iota \circ \mu^{n}_{\tau,\pi} (x) $,
which is equivalent to
$ (\varphi^{-1} \circ \psi ) (x) \cap ( (\alpha^{n}_{\tau,\pi})^{-1} \circ   \mu^{n}_{\tau,\pi} )(x)
\neq \emptyset$,
and means that
at least one forest $r$ exists that can be inserted at the port
in a consistent way, i.e. such that
$\varphi(r) = \psi(x)$ and $\alpha^{n}_{\tau,\pi}(r) = \mu^{n}_{\tau,\pi} (x)$.
Combined with $m$, the mappings $\mu^{n}_{\tau,\pi}$
and $\psi$ define
the tuple  $t = (m,\mu^{n}_{\tau,\pi},\psi)$;
it is said to be consistent if the two mappings are 
consistent at every port of $m$.
The leaf-completion of $m$ through 
each of these  mappings will be used at several places.
%
%
%
%
\newline
%
%
We establish a relation between two consistent tuples
 $t=(m,\mu^{n}_{\tau,\pi},\psi)$ and $t'= (m',\mu^{n}_{\tau,\pi},\psi)$.
Let $\Brvmu^{n}_{\tau,\pi}(m) $, a forest over $A \cup H^{n}_{\tau,\pi} $,
be the leaf-completion of $m$ through $\mu^{n}_{\tau,\pi} (x)$;
define on it the leaf extension of
$\alpha^{n}_{\tau,\pi}$.
Then in $\BbbF_{A  \cup H^{n}_{\tau,\pi} }$,
the nuclear congruence of $\alpha^{n}_{\tau,\pi}$ coincides with 
the congruence 
$\approx^{n+1}_{\tau,\pi} $.
Next, with
$ h \in H^{n}_{\tau,\pi} $, $ v \in V^{n}_{\tau,\pi} $ and $ j \in G $,
let
$P[\calH^{n}_{\tau,\pi}](t)_{h,v,j}$ denote the number of
ports $x \in ports(m)$  that satisfy
$h = \mu^{n}_{\tau,\pi}(x)  $,
$v = \alpha^{n}_{\tau,\pi} (\nabla (\Brvmu^{n}_{\tau,\pi}(m),x))  $
and
$j = \psi(x)    $.
From this we define 
$P[\calH^{n}_{\tau,\pi}](t)  \equiv_{\tau,\pi} 
P[\calH^{n}_{\tau,\pi}](t')$
 as a shorthand for
\[
\forall\, \crochet{ h,v,j}  \in H^{n}_{\tau,\pi} \times V^{n}_{\tau,\pi} \times F :\
P[\calH^{n}_{\tau,\pi}](t)_{h,v,j} \ \equiv_{\tau,\pi} \
P[\calH^{n}_{\tau,\pi}](t')_{h,v,j} .
\]
Finally, we define the relation $\crochet{\star \calH^n _{\tau,\pi}  }$ between tuples:
\[
t\ \crochet{\star \calH^n _{\tau,\pi}  } \  t'
 \ \Leftrightarrow\
\Brvmu^{n}_{\tau,\pi}(m)  \approx^{n+1}_{\tau,\pi} \Brvmu^{n}_{\tau,\pi}(m') 
\ \wedge\
P[\calH^{n}_{\tau,\pi}](t) \ \equiv_{\tau,\pi} \
P[\calH^{n}_{\tau,\pi}](t')  .
\]

We now describe the meaning of this relation.
%
We want to determine conditions that are sufficient for two tuples
$t=(m,\mu^{n}_{\tau,\pi},\psi)$ and $t'= (m',\mu^{n}_{\tau,\pi},\psi)$
to be suitable for the construction of witnesses $s$ and $s'$
such that
$s   \approx^{n+1}_{\tau,\pi} s'$, while
$\varphi(s) $ and $\varphi(s')$ are distinct elements of the
same $\approx^\N_{\tau,\pi}$-tight subset of $G$.
The forests  $s$ and $s'$ are built from
 $t$ and $t'$ by inserting at their ports elements of
a set $\calS^{(n)}$ of witnesses for $\approx^{n}_{\tau,\pi}$;
the insertion at port $x$ of
the forest $s(x)$ is consistent, in the sense defined above.
 Let
$\tilde{s}$ and  $\tilde{s}'$ denote the relabeled versions
of $s$ and $s'$, relative to
$\alpha^{n}_{\tau,\pi}$, and let
$A^n = A \times H^{n}_{\tau,\pi} \times V^{n}_{\tau,\pi}$
be the corresponding relabeling alphabet.
We assume that $\varphi(s)$ and $\varphi(s')$
have the appropriate values and we look at how to
satisfy
the constraint 
$s   \approx^{n+1}_{\tau,\pi} s'$.
Satisfaction is obtained when every symbol of
$ A^n$ occurs the same number of times
(up to $\equiv_{\tau,\pi}$) in $\tilde{s}$ and  $\tilde{s}'$.
We verify  this  separately on the nodes
of the multicontexts
of $m$ and $m'$ and on those that belongs to the
copies of elements of $\calS^{(n)}$
that are inserted at their ports.
Given $x \in ports(m)$, we have by construction
\[
\alpha^{n}_{\tau,\pi} (\Delta^+(s,x) ) =
\alpha^{n}_{\tau,\pi} (\Delta(s,x) ) =
\alpha^{n}_{\tau,\pi} (s(x) ) =
\mu^{n}_{\tau,\pi} (x) 
\]
where the leftmost equality comes from the fact that in
$s$, the label of $x$ is the neutral letter $e$.
Therefore, for every interior node $y$ of $m$,
\[
\alpha^{n}_{\tau,\pi} (\Delta(s,y) ) =
\alpha^{n}_{\tau,\pi} (\Delta(   \Brvmu^{n}_{\tau,\pi} (m)  ,y) ) 
\quad \text{and} \quad
\alpha^{n}_{\tau,\pi} (\nabla(s,y) ) =
\alpha^{n}_{\tau,\pi} (\nabla(   \Brvmu^{n}_{\tau,\pi} (m)  ,y) ) 
\]
which explains the component
$
\Brvmu^{n}_{\tau,\pi}(m)  \approx^{n+1}_{\tau,\pi} \Brvmu^{n}_{\tau,\pi}(m') 
$
in the definition of
$t\ \crochet{\star \calH^n _{\tau,\pi}  } \  t'
$.
%
Next, we deal with the copies of witnesses from $\calS^{(n)}$
inserted at the ports of $m$ and $m'$.
Define as above the counter
$P[\calH^{n}_{\tau,\pi}](t)_{h,v,j}$
and the shorthand
$P[\calH^{n}_{\tau,\pi}](t)  \equiv_{\tau,\pi} 
P[\calH^{n}_{\tau,\pi}](t')$.
The witness $s(x)$ is a copy of $s_{h,j}^{(n)}$;
let $y$ be one of its nodes,
with 
$a = \lambda(s_{h,j}^{(n)},y)$,
we have 
\[
\alpha^{n}_{\tau,\pi} ( \Delta ({s},y) ) =  \alpha^{n}_{\tau,\pi} ( \Delta (s_{h,j}^{(n)},y) ) = k
\]
and
\[
\alpha^{n}_{\tau,\pi} (\nabla ({s},y) ) = \alpha^{n}_{\tau,\pi} ( \nabla ({s},x) {\cdot} \nabla(s_{h,j}^{(n)},y) ) 
= v {\cdot} \alpha^{n}_{\tau,\pi} (\nabla (s_{h,j}^{(n)},y)) = vw,
\]
so that
$\lambda(\tilde{s},y) = \crochet{a,k,vw}$.
If $P[\calH^{n}_{\tau,\pi}](t)  \equiv_{\tau,\pi} 
P[\calH^{n}_{\tau,\pi}](t')$, then for every combination of
$s_{h,j}^{(n)} \in \calS^{(n)}$ and  $ v \in V^{n}_{\tau,\pi} $,
the numbers of copies of  $s_{h,j}^{(n)}$ that are
inserted at ports $x$ such that 
$v = \alpha^{n}_{\tau,\pi} (\nabla (\Brvmu^{n}_{\tau,\pi}(m),x))  $
are the same in $m$ and $m'$, up to $\equiv_{\tau,\pi}$.
%


%
%

\begin{definition} \label{def:RC}
We denote by 
$\mathbf{RC}(\calG {\star} \calH^n _{\tau,\pi})$ 
the existence of a
circuit $\calM^{(n+1)} $ over
$(A,\calD^{(n)},\calD^{(n+1)})$ such that
%
\begin{latindense}
\item every tuple $t^{(n+1)}_{h,j}  = (m,\mu^n_{\tau,\pi} ,\psi)$ 
is consistent and satisfies
$\alpha^{n+1}_{\tau,\pi} (\Brvmu^{n}_{\tau,\pi}(m) ) = h$ and
$\varphi(\Brvpsi(m)) = j$, and 
\item   every combination of
 $t^{(n+1)}_{h,j}$ and $t^{(n+1)}_{h,j'} $ satisfies  condition
%
$t^{(n+1)}_{h,j}\ \crochet{\star \calH^n _{\tau,\pi}  }\ t^{(n+1)}_{h,j'} $. %
\end{latindense}
\end{definition}

%


\begin{lemma} \label{lem:preRC}
Let $n \ge 2$: if a  circuit $\calM^{(n+1)} $
satisfies $\mathbf{RC}(\calG  {\star} \calH^{n}_{\tau,\pi})$,
then it also satisfies $\mathbf{RC}(\calG  {\star} \calH^{n-1}_{\tau,\pi})$,
and can be used to prove
 $\calG \not\in  {{\mathbf{**}}^{n+1}}\bicrochet{\N_{\tau,\pi}}$.
\end{lemma}

\begin{proof}
Let  $\calM^{(n+1)}$ satisfy $\mathbf{RC}(\calG  {\star} \calH^{n}_{\tau,\pi})$.
Define  the surjective homomorphism $\sigma : \calH^{n}_{\tau,\pi} \rightarrow \calH^{n-1}_{\tau,\pi}$
 by $\sigma = \alpha^{n-1}_{\tau,\pi} \circ (\alpha^{n}_{\tau,\pi} )^{-1}$.
Then from every consistent tuple $(m, \mu^n_{\tau,\pi}, \psi) $ 
in $\calM^{(n+1)}$ one can build 
a consistent $(m, \mu^{n-1}_{\tau,\pi} ,  \psi) $,
where
$ \mu^{n-1}_{\tau,\pi} = \sigma \circ \mu^n_{\tau,\pi}$.
We define from 
$(m, \mu^{n-1}_{\tau,\pi}, \psi) $
its leaf completion
$\Brvmu^{n-1}_{\tau,\pi}(m)$;
we have
\[
\Brvmu^{n}_{\tau,\pi}(m)  \approx^{n}_{\tau,\pi} \Brvmu^{n}_{\tau,\pi}(m') 
\Rightarrow
\Brvmu^{n-1}_{\tau,\pi}(m)  \approx^{n-1}_{\tau,\pi} \Brvmu^{n-1}_{\tau,\pi}(m') 
\]
 and,
for every combination of $h \in H^n_{\tau,\pi}$,
$v \in V^n_{\tau,\pi}$ and $j \in G$, 
the counters 
$P[\calH^{n-1}_{\tau,\pi}](t)_{h,v,j}$  satisfy
\[
P[\calH^{n-1}_{\tau,\pi}](t)_{h,v,j} =
\sum_{\substack{\sigma(k)=h \\ \sigma(w)=v}}
P[\calH^{n}_{\tau,\pi}](t)_{k,w,j} .
\]
Therefore, if $\calS^{(n-1)}$ is a vector of witnesses for
$\approx^{n-1}_{\tau,\pi}$, then
$\calM^{(n+1)} {\cdot} \calS^{(n-1)}$ is a vector of witnesses for
$\approx^n_{\tau,\pi}$.
Note that since
$|H^{n+1}_{\tau,\pi}| > |H^{n}_{\tau,\pi}|$,
some pairs $(h,j) \in \calD^{(n)}$
end up with more than one witness.
\newline
Sets of witnesses
$\calS^{(1)} ,\ldots, \calS^{(n)}$
can be built from a
circuit $\calM^{(n+1)}$ 
that  satisfies 
$\mathbf{RC}(\calG  {\star} \calH^{n}_{\tau,\pi})$,
as follows.
First, select
a set $\calS^{(0)}$ that contains a forest
$s_{J,j}$ for every $(J,j)$ and $h$ 
with $J = \iota(h)$ and $(h,j) \in \calD^{(n)}$,
such that $\varphi(s_{J,j}) = j$,
and build 
$\calS^{(1)} = \calM {\cdot} \calS^{(0)}$:
for any two tuples
$t^{(n+1)}_{h,j}$ and $t^{(n+1)}_{h,j'}$
and the corresponding forests
$s^{(1)}_{h,j}$ and $s^{(1)}_{h,j'}$,
having
$t^{(n+1)}_{h,j}\, (\calG \star \calH^n _{\tau,\pi}  )\, t^{(n+1)}_{h,j'} $
ensures that every $a \in A$ occurs the same number of times
(up to $\equiv_{\tau,\pi}$) in $s^{(1)}_{h,j}$ and $s^{(1)}_{h,j'}$,
so that 
$\calS^{(1)} $ constitutes a set of witnesses for
$\approx^1_{\tau,\pi}$.
Then recursively,
$\calS^{(n)} = \calM {\cdot} \calS^{(n-1)} = \calM^n {\cdot} \calS^{(0)}$,
for every $n \ge 2$.
\end{proof}

\begin{lemma} \label{lem:secondRC}
$\calG \not\in {{\mathbf{**}}}\bicrochet{\N_{\tau,\pi}}$
if, and only if there exists a sequence of full circuits $\calM^{(n+1)} $, $n \ge 1$,
that satisfies $\mathbf{RC}(\calG  {\star} \calH^{n}_{\tau,\pi})$.
\end{lemma}

\begin{proof}
By the above discussion, satisfaction of
$\mathbf{RC}(\calG \star \calH^{n}_{\tau,\pi})$ by every $\calM^{(n+1)}$, $n \ge 1$, 
ensures that this sequence of  circuits 
can be used to build a recursive proof. 
For the \emph{only if} direction,
assuming  $\calG \not\in {{\mathbf{**}}}\bicrochet{\N_{\tau,\pi} }$,
we want to prove the existence of a sequence of circuits where each
$\calM^{(n+1)}$, $n \ge 1$,  satisfies $\mathbf{RC}(\calG \star \calH^{n}_{\tau,\pi})$.
By Lemma \ref{lem:preRC}, it suffices to prove this for $n$ above a
large enough threshold: we select for this
the parameter $n_0$ defined at the end of Section \ref{sec:prelim}.
%
%
\newline
The proof uses the algebra
$\calG' = \calG \times \calH^{n}_{\tau,\pi} $;
since
$\calG \not\in {{\mathbf{**}}}\bicrochet{\N_{\tau,\pi}} $, 
we have
$\calG' \not\in {{\mathbf{**}}}\bicrochet{\N_{\tau,\pi}} $.
With the notations
$\calG' = (G',W') $
and $\Gout' = \{ g' \in G' : L_{g'} \not\in  {{\mathbf{**}}}\bicrochet{\N_{\tau,\pi}} \}$,
we observe that every
$\approx_{\tau,\pi}^\N$-tight subset of  $\Gout'$
can be written $J' = J \times \{h\}$,  
where $J \subseteq \iota(h)$ and
$\iota(h)$ is $\approx^\N_{\tau,\pi}$-tight.
To see this, consider  two elements $(j,h)$ and $(k,\ell) $
of $J'$. 
If $h \neq \ell$, then $\calH^{n}_{\tau,\pi} $
can tell apart the languages $L_{(j,h)}$ and $L_{(k,\ell) }$,
which means that $J'$ is not  $\approx^\N_{\tau,\pi}$-tight.
Next, if  $h = \ell$ 
but $j$ and $k$ do not belong to the same
$\approx_{\tau,\pi}^\N$-tight
subset of \Gout, then by 
Proposition \ref{prp:Nzero},
 $ \alpha^{n}_{\tau,\pi} $ maps the
 sets $\varphi^{-1}(j)$ and $\varphi^{-1}(k)$
onto distinct subsets 
(actually, singletons)
of $H^{n}_{\tau,\pi} $.
Meanwhile, given
$h \in H^{n}_{\tau,\pi} $, a contradiction argument shows that if
$\iota(h)$ is a $\approx^\N_{\tau,\pi}$-tight
subset of \Gout, then  $\iota(h) \times \{h\}$ is also $\approx^\N_{\tau,\pi}$-tight.
\newline
%
%
%
The proof consists in showing that any full vector of witnesses
for $\approx^{n+1}_{\tau,\pi}$ can be used to define a full circuit
$\calM^{(n+1)} $ that satisfies $\mathbf{RC}(\calG \star \calH^{n}_{\tau,\pi})$.
Let $s$ be such a witness; we  build 
a multicontext $m$ from $s$ as follows.
In $m$, every leaf is a port; the interior nodes,
i.e. the set $interior(m)$,
are exactly  those $y \in nodes(s)$
for which $\Delta(s,y)$ contains at least one pathhead 
for $\Gout'$; in other words, every  strict
ancestor of a pathhead
becomes
an interior node of $m$.
Next, we insert a port $x$ along every edge $(y,z)$ where
$ y \in interior(m)$ and $z \in nodes(s) - interior(m)$,
and we remove the subtree of $s$ rooted at $z$.
We will prove that 
the resulting multicontext has the required properties.
Let $y \in nodes(s)$ and let $\tilde{s}$ denote the version of $s$ relabeled
relative to $\alpha^{n}_{\tau,\pi} $: we show that the triple
$\lambda(\tilde{s},y) = 
\crochet{
\lambda(s,y) ,\, 
\alpha^{n}_{\tau,\pi} (\Delta (s,y)) ,\,
\alpha^{n}_{\tau,\pi} (\nabla (s,y)) }$
can be used to
determine whether $y$ is a pathhead for $\Gout'$.
Let \Sstar\ denote the set of all trees over $A$,
\calJ\ the set of all $\approx^\N_{\tau,\pi}$-tight
subsets of \Gout,
and define
\[
H_{\Tree} = \{\, h \in H^{n}_{\tau,\pi} : 
\iota(h) \in \calJ \wedge (\alpha^{n}_{\tau,\pi})^{-1}(h) \subset \Sstar \, \}  .
\]
Since \Sstar\  is recognized by
every algebra
$\calH^{2}_{\tau,\pi} $ where $\tau \ge 1$,
each set $(\alpha^{n}_{\tau,\pi})^{-1}(h)$ is either a subset of
\Sstar\ or disjoint with it.
Then the set of all forests that contain
a pathhead for $\Gout'$ is recognized
by $ \calH^{n}_{\tau,\pi}$ through
$ \alpha^{n}_{\tau,\pi}$
and the  set $H_{\text{ph}}$ of all elements of
$H^{n}_{\tau,\pi}$
accessible from
$H_{\Tree}$.
A  node $y$ is a pathhead 
for $\Gout'$ when 
$\alpha^{n}_{\tau,\pi} (\Delta^+ (s,y)) $
belongs to $H_{\text{ph}}$ and 
$\alpha^{n}_{\tau,\pi} (\Delta (s,y)) $
does not.
The nodes of $m$ are those $y \in nodes(s)$
for which $\alpha^{n}_{\tau,\pi} (\Delta (s,y)) \in H_{\text{ph}}$, 
that is, $y$ is a strict ancestor
of at least one pathhead.
To build the ports of $m$, 
we take each father-son pair $y,z$
in $s$
where $y \in interior(m)$ and $z \not\in  interior(m)$,
and insert a port $x$  between the two,
so that
$\Delta (s,x) = \Delta^+ (s,z)$
and also
$\beta(\nabla(s,x)) = \beta(\nabla(s,z)) $
for every homomorphism $\beta$.
By construction, $m$ and $x$ satisfy
$
\alpha^{n}_{\tau,\pi} (\nabla (\Brvmu^{n}_{\tau,\pi}(m),x))
= 
\alpha^{n}_{\tau,\pi} (\nabla (s,z)) $,
$
\alpha^{n}_{\tau,\pi} (\Brvmu^{n}_{\tau,\pi}(m))
= 
\alpha^{n}_{\tau,\pi} (s) $
and
$\varphi(\Brvpsi(m)) = \varphi(s)$.
We also observe that 
$\varphi(\Delta (s,x)) = \varphi(\Delta^+ (s,z))$
is determined by $\alpha^{n}_{\tau,\pi} (\Delta^+ (s,z)) $.
Indeed, either $z$ is not a pathhead and the set
$\iota(\alpha^{n}_{\tau,\pi} (\Delta^+ (s,z)) )$
is not  $\approx^\N_{\tau,\pi}$-tight,
or $z$ is a pathhead, a case 
discussed in the proof of Proposition \ref{prp:cfc}.
Hence, 
the value of $\lambda(\tilde{s},z) $
determines both
$\mu^{n}_{\tau,\pi}(x) = \alpha^{n}_{\tau,\pi} (\Delta^+ (s,z)) $
and
$\psi(x) = \varphi(\Delta^+ (s,z))$;
from there, it determines which counter
$P[\calH^{n}_{\tau,\pi}](t)_{h,v,j}$ the port $x$
contributes to.
Given a vector of integers $\vec{P}$,
the algebra $\calH^{n+1}_{\tau,\pi}$
can therefore recognize whether
$P[\calH^{n}_{\tau,\pi}](t)_{h,v,j} \equiv_{\tau,\pi} \vec{P}_{h,v,j}$.
As a consequence, given another witness $s'$
and the multicontext $m'$ built from $s'$,
if $s' \approx^{n+1}_{\tau,\pi}  s$, then
the corresponding tuples 
$t$ and $t'$ 
satisfy $P[\calH^{n}_{\tau,\pi}](t) \equiv_{\tau,\pi} 
P[\calH^{n}_{\tau,\pi}](t') $.
%
%
Finally, the fact that the interior nodes of $m$ are exactly those  $y \in nodes(s)$
for which $\Delta(s,y)$ contains a pathhead,
means that the labels they carry in $\tilde{s}$ are distinct from those of
the other nodes of $s$.
Therefore, defining 
$\mu^{n}_{\tau,\pi} (x) = \alpha^{n}_{\tau,\pi} (\Delta(s,x))$
for every port of $m$, we obtain
$\Brvmu^{n}_{\tau,\pi} (m) \approx^{n+1}_{\tau,\pi}  \Brvmu^{n}_{\tau,\pi} (m')$.
Defining further $\psi(x) = \varphi ( \Delta(s,x))$, we build tuples that
are equivalent under
$\crochet{ \star \calH^n_{  \tau,\pi} }$
and thus
constitute a circuit that satisfies  $\mathbf{RC}(\calG  {\star} \calH^{n}_{\tau,\pi})$.
\end{proof}
\vspace{0.1in}

The number of forests in a full vector of witnesses $\calS^{(n)}$
increases with $n$, while the existing examples of proofs
describe sequences $\calS^{(n)}$, $n \ge 1$, where all vectors have
the same size. We say that a proof is 
\emph{slender} when,
for every combination of a $\approx^\N_{\tau,\pi}$-tight
set $J$ and of $j \in J$, it contains
at most one witness $s^{(n)}_{h,j}$ such that $\iota(h) = J$.

\begin{theorem} \label{thm:RC}
We have $\calG \not\in {{\mathbf{**}}}\bicrochet{\N_{\tau,\pi}}\ $ iff
~\calG\
has a slender recursive proof of non-membership
which for each $n$, 
involves a circuit that satisfies Condition $\mathbf{RC}(\calG  {\star} \calH^{n}_{\tau,\pi})$.
\hfill $\square$
\end{theorem}

\begin{proof}
It suffices to prove the \emph{only if} direction.
Recall that \calJ\ denotes the set of all
$\approx^\N_{\tau,\pi}$-tight subsets of $G$.
Given $n \ge 1$, we define 
$\delta_n(\calJ) = \{ \, h \in H^n_{  \tau,\pi} : \iota(h) \in \calJ \, \}$,
we pick a  representative of every class
for the equivalence relation
$\delta_n \circ \iota(h)$, 
and define 
$\rho_n : \delta_n(\calJ)  \rightarrow \delta_n(\calJ)$
which maps every $h \in \delta_n(\calJ) $ to its representative.
\newline
Let $t = (m, \mu^{n}_{\tau,\pi} ,  \psi )$ and 
$t' = (m', \mu^{n}_{\tau,\pi} ,  \psi )$ be two tuples in 
$\calM^{(n+1)}$, let $s$ and $s'$ denote the witnesses
built from them, with
$s \approx^{n+1}_{\tau,\pi}  s'$ and
$\varphi(s) \neq \varphi(s')$.
In order to use $m$  in
a slender proof, 
we modify the labeling of  every  $x \in ports(m) $,
we replacing $\mu^{n}_{\tau,\pi} (x)$
with
$  \rho_n(\mu^{n}_{\tau,\pi} (x))$.
We then build from $m$ a forest  $s_\rho$ 
by inserting at every port $x$ a copy of $s^{(n)}_{\rho_n(\mu^{n}_{\tau,\pi} (x)),\psi(x)}$,
(instead of $s^{(n)}_{\mu^{n}_{\tau,\pi} (x),\psi(x)}$);
the construction of
$s_\rho$  uses a slender subset 
$\calS_\rho^{(n)}$ of $\calS^{(n)}$,
where $s^{(n)}_{h,j} \in \calS_\rho^{(n)}$
iff $h = \rho_n(h)$.
By construction, we have
$\varphi(s_\rho) = \varphi(s)$.
We build 
$s_\rho'$ from $m'$ in the same way.
Verifying that
$s_\rho \approx^{n+1}_{\tau,\pi}  s_\rho'$
will lead us to conclude that
$s_\rho $ and $ s_\rho'$ are witnesses
for the same $\approx^\N_{\tau,\pi}$-tight set as $s$ and $s'$.
This is done by induction on $n$,
proving for every $i \le n$ and every pair 
$x,x' $ of nodes or ports in $m$ and $m'$ that if
$\Delta (s,x) \approx^i_{\tau,\pi} \Delta(s',x')$
and
$\nabla (s,x) \approx^i_{\tau,\pi} \nabla(s',x')$,
then
$\Delta (s_\rho,x) \approx^i_{\tau,\pi} \Delta(s_\rho',x')$
and
$\nabla (s_\rho,x) \approx^i_{\tau,\pi} \nabla(s_\rho',x')$
hold as well.
As a consequence, those nodes of $m$
and $m'$ that carry the same label (element of $A^n$)
in the relabeled versions of $s$ and $s'$
will also carry the same label (but possibly not the original one)
in the relabeled versions of $s_\rho$ and $s_\rho'$;
the tuples $t_\rho$ and $t_\rho'$ built from them
are equivalent under 
$\crochet{ \star \calH^n_{  \tau,\pi} }$.
Therefore, with this method one can start with a sequence $\calS^{(n)}$,
$n \ge 1$, of full
sets of witnesses and build from it a slender proof.
\end{proof}
\vspace{0.08in}

If $\calI $ is a  subset 
of \calL\ such that  $\bigcup_{J \in \calI} J$
is strongly connected
and is maximal relative to this property, then  
the witnesses for \calI\ 
constitute by themselves a recursive proof.
We can therefore from now restrict our work on 
proofs that involve sets \calL\ where  $\bigcup_{J \in \calL} J$
is strongly connected.
\newline
Let 
$\mathbf{NRC}(\calG {\star} \calH^n_{\tau,\pi})$
denote the negation of  
$\mathbf{RC}(\calG {\star} \calH^n_{\tau,\pi})$;
by Lemma \ref{lem:preRC},
it
states
that there does not exist witnesses for
\calG\ relatively to $\calH^n_{\tau,\pi}$.
It
is tempting to try to obtain from it a  simpler formula,
such as 
$\forall\, t,t' : t\, \crochet{\star \calH^n_{\tau,\pi}}\, t' \Rightarrow \varphi(t) = \varphi(t')$,
which is  an equational definition for a variety of forest algebras.
This does not work, because
$\mathbf{RC}(\calG {\star} \calH^n_{\tau,\pi})$
applies only within the strongly connected components of $G$.
For example,
consider the language $L$ over $A=\{a,b\}$
defined by
\emph{at least one node is an ancestor of exactly one node labelled $b$},
whose syntactic algebra belongs to 
$\mathbf{**}^2\bicrochet{\N_{2,1}}$, and the tuples $t=\Brvpsi(m)$ and $t'=\Brvpsi'(m)$ built over the 
multicontext
$m = a(a(x_1+x_2)+a(x_3+x_4))$ with three interior nodes (labelled $a$)
and four ports 
where the mappings $\psi$ and $\psi'$ are defined  by
$\psi(x_1) = \psi(x_3) = \psi'(x_1) = \psi'(x_2) = a$
and $\psi(x_2) = \psi(x_4) = \psi'(x_3) = \psi'(x_4) = b$,
so that $\Brvpsi(m) \in L$, $\Brvpsi'(m) \not\in L$,
and
$t\, \crochet{\star \calH^2_{2,1}}  \,t'$.
%

%



\subsection{Uniform proofs}   \label{sec:uniproof}

Recall that $\stackrel{\sigma,\rho}{\leftrightarrow}$
denotes the 
equivalence-under-pumping  congruence in $\BbbF_{A,B}$,  
and $\calJ^{\sigma,\rho} = ( J^{\sigma,\rho}, U^{\sigma,\rho})$ the quotient algebra.
We use the same notations for the corresponding objects defined over \Sfree.
Given a tuple
${t} = (m,\nu,\psi)$, for
every combination of
$v \in U^{\sigma,\rho}$ and $(J,j) \in \calD$ we define  
the counter
\[
P[\calJ^{\sigma,\rho}]( {t} )_{v,J,j} \ =\ 
| \{ x \in ports(m) : \beta^{\sigma,\rho}(\nabla(m,x)) = v, \, \nu(x) = J, \, \psi(x) = j \} |  .
\]
From this, we define a relation in $\BbbM_{A,D}$:
\[
 {t} \, \crochet{ \star  \calJ^{\sigma,\rho} } \,  {t}'
\ \Leftrightarrow\
( m \stackrel{\sigma,\rho}{\leftrightarrow} m' )
\wedge
\left(
\forall\, v \in U^{\sigma,\rho} ,\ \forall\, (J,j) \in \calD :\
P[\calJ^{\sigma,\rho}](t)_{v,J,j} \ \equiv_{\tau,\pi} \
P[\calJ^{\sigma,\rho}](t')_{v,J,j} 
\right) .
\]
In a circuit $\calM^{(n+1)}$ found in a proof-by-pumping,
every subcircuit $T^{(n+1)}_J$ is diagonal and closed under
$\crochet{ \star  \calJ^{\sigma,\rho} } $: we denote by 
$\mathbf{RC}(\calG  {\star}  \calJ^{\sigma,\rho} )$
the proposition that states the existence of such a circuit
$\calM^{(n+1)}$.
Since
$\stackrel{\sigma,\rho}{\leftrightarrow}$
refines $ \approx^{n+1}_{\tau,\pi}$,
the relation
$\crochet{ \star  \calJ^{\sigma,\rho} } $ refines  
$\crochet{ \star \calH^n_{\tau,\pi} } $, and therefore
$\mathbf{RC}(\calG  {\star}  \calJ^{\sigma,\rho} )$
implies
$\mathbf{RC}(\calG  {\star}  \calH^n_{\tau,\pi})$.
\newline
%
%
Consider an algebra
$\calG \in {\mathbf{**}}\bicrochet{\N_{\tau,\pi} }$ and
let $\sigma$ and  $\rho$ be  the  threshold and period  of $\calG_\%  $.
This means that \calG\
belongs to the variety
${\mathbf{**}}\bicrochet{\N_{\tau,\pi} }\! \wedge\! \bicrochet{\calJ^{\sigma,\rho}}$,
where
$\bicrochet{\calJ^{\sigma,\rho}}$ is the variety of 
finite forest algebras
generated by $\calJ^{\sigma,\rho}$.
We can repeat with the congruences 
$\left( \approx^n_{\tau,\pi} \cap   \stackrel{\sigma,\rho}{\leftrightarrow} \right) $,
for every $n \ge 1$, the work done for Proposition \ref{prp:basic}\ and
Theorem \ref{thm:RC}\ and assert the existence of an algebra
$\calK^{n,\sigma,\rho}_{\tau,\pi} = \Sfree/\left({ \approx^n_{\tau,\pi} \cap   \stackrel{\sigma,\rho}{\leftrightarrow}}\right)$
such that 
$\calG \in {\mathbf{**}}\bicrochet{\N_{\tau,\pi} } \!\wedge\! \bicrochet{\calJ^{\sigma,\rho}}$
if, and only if
$\mathbf{NRC}(\calG  {\star}  \calK^{n,\sigma,\rho}_{\tau,\pi})$
for some $n \ge 1$.
%
%
Since
$\calK^{n,\sigma,\rho}_{\tau,\pi} \prec \calJ^{\sigma,\rho}$,
this implies
$\mathbf{NRC}(\calG  {\star}   \calJ^{\sigma,\rho})$.
We obtain the following.

\begin{proposition}  \label{prp:Jsigma}
With  $\sigma$ and $\rho$
the threshold and period of $\calG_\%$,
if $ \calG \in
{\mathbf{**}}\bicrochet{\N_{\tau,\pi} }$,
then
$\mathbf{NRC}(\calG  {\star}   \calJ^{\sigma,\rho})$.
\end{proposition}

%
The values of $\sigma$ and $\rho$ 
are computable in finite time; 
with them in hand,
determining the existence of a
circuit \calT\ that satisfies
$\mathbf{RC}(\calG  {\star}   \calJ^{\sigma,\rho})$
is a recursively enumerable problem.
If the corresponding algorithm stops,
then 
$\calG \not\in {\mathbf{**}}\bicrochet{\N_{\tau,\pi} }$.
However,  it
is not clear whether 
a whole proof can be built from \calT,
that is, whether for  every  $n \ge 2$, 
knowledge of \calT\ 
suffices to determine
every set of witnesses $\calS{(n)}$.
If so, then this would be a uniform proof of non-membership.
%
%
The results of Section \ref{sec:RP}\ make no mention
of uniformity. However,
the existing Ehrenfeucht-Fra\"{\i}ss\'e games are uniform
and  are built along
two construction mechanisms,
which we call
proof-by-copy and proof-by-pumping,
where the former is a special case of the latter.


\subsubsection{Proof-by-copy}   \label{sec:pbc}

In certain cases such as the Boolean algebra (see Section \ref{sec:boole}),
condition
$\mathbf{RC}(\calG  {\star}  \calH^n_{\tau,\pi})$
can be satisfied
with a circuit $\calT$ where each subcircuit $T_J$ is 
built over the same multicontext, 
 that is,
$m_{J,j} \stackrel{\infty}{\leftrightarrow} m_{J,j'} $
for all $j,j' \in J$.
Then \calT\ actually satisfies
$\mathbf{RC}(\calG  {\star}  \calH^n_{\tau,\pi})$
for every $n \ge 1$ and can be used in
the construction of every set of witnesses
$\calS^{(n)}$, in the way described in the
proof of Lemma \ref{lem:preRC},
and
proving 
$\calG \not\in {\mathbf{**}}\bicrochet{\N_{\tau,\pi} }$
is just a matter of creating copies of \calT\
and assembling them\footnote{Observe that
the depth of this construction increases linearly with $n$.} into witnesses,
hence the name ``proof-by-copy''.
We verify that the existence of such a proof is a
recursively enumerable problem.
%
%
\newline
Formally, we work with $\calG = (G,W) = \varphi(\Sfree)$, 
a connected component $F$, the set
$\calD = \{ (J,j) : \emptyset \neq J \subseteq F,\, j \in J \}$,
and, to make things simpler, the assumption that
$F$ is either the minimal ideal of $G$,
i.e.
$wF \subseteq F$ for every $w \in W$,
or $F \cup \{\infty\}$ is an ideal,
where $w\infty = \infty$ for every $w \in W$.
Given $J \subseteq F$, we define
$J_\infty = J \cup \{\infty\}$.
We use
tuples $t = (m,\nu,\psi) \in \BbbM_{A,D}$
with  $(\nu,\psi) : ports(m) \rightarrow \calD$.
Since
$ \stackrel{\infty}{\leftrightarrow} $
coincides with $=$, up to horizontal permutation of siblings,
we confound $\Sfree$ with $\Sfree / {\stackrel{\infty}{\leftrightarrow}} $
and the canonical homomorphism with the identity mapping.
The \emph{support} of a  tuple $t$ is
\[
sup(t) = sup(m) = \{  \nabla(m,x) : x \in ports(m) \} ;
\]
its \emph{set of inputs} is
$inp(t) = inp(m) = \{ J \subseteq F : \exists\, x \in ports(m),\, \nu(x) = J \}$.
Next, for
every combination of
$v \in \Svert$ and $(J,j) \in \calD$ we define  
the counter
\[
P[\Sfree](t)_{v,J,j} \ =\ 
| \{ x \in ports(m) : \nabla(m,x) = v, \, \nu(x) = J, \, \psi(x) = j \} |  ,
\]
with $P[\Sfree](t)_{v,J,j} = 0$
whenever $J \not\in inp(t)$
or
$v \not\in sup(t)$.
Finally, we define a relation in $\BbbM_{A,D}$:
\[
t\, \crochet{ \star  {}_{\tau,\pi}}\, t'
\ \Leftrightarrow\
( m \stackrel{\infty}{\leftrightarrow} m' )
\wedge
\left(
\forall\, v \in \Svert ,\ \forall\, (J,j) \in \calD :\
P[\Sfree](t)_{v,J,j} \, \equiv_{\tau,\pi} \,
P[\Sfree](t')_{v,J,j} 
\right) .
\]
A proof-by-copy for \calG, if it exists, can be found by an
algorithm that traverses the set of all slender
circuits \calT\ over subsets of $F$,
and stops when it has found one where every 
subcircuit $T_J$ is diagonal and closed under
$\crochet{ \star  {}_{\tau,\pi}}$.


%
%

 \subsubsection{Proof-by-pumping} \label{sec:pbp}

Other Ehrenfeucht-Fra\"{\i}ss\'e games   
are sequences of witnesses 
where each subcircuit 
of $\calT^{(n+1)}$  is 
built over  multicontexts that are 
equivalent under an
``equivalence under pumping'' congruence
$\stackrel{\sigma,\rho}{\leftrightarrow} $
that refines
$\approx^{n+1}_{\tau,\pi}$;
hence the  name ``proof-by-pumping''.
In the examples  described in 
the next section,
 pumping is used in two different ways.
In Section \ref{sec:evendepth},
the circuits consists in uniform multicontexts
where all ports belong to the same set $Z$
and that are obtained from
a single
$m \in \BbbM_{A,B}$ by pumping along $Z$.
In Section  \ref{sec:potthoff}, work starts with a unique
multicontext  whose ports are partitioned into two
sets $Y$ and $Z$; pumping  $\theta$ times along 
$Z$
creates a multicontext denoted $p^{(1)}_\theta$;
copies of $p^{(1)}_\theta$ and $p^{(1)}_{\theta+1}$
are then assembled, by insertion at the $Y$-ports,
into the actual components of the circuit
used in the construction of the witnesses.
%
%
In this section, we investigate the first of
these two techniques; then we examine another
situation where pumping is involved,
but no example of which is known to the author.
\newline
An algorithm adapted from the one
described at the end of Section \ref{sec:pbc}\
can verify  the existence of a 
suitable \calT\ where every subcircuit 
$T_J$ is diagonal and closed under
$\crochet{ \star  \calJ^{\sigma,\rho} } $,
for a given pair $\sigma,\rho$.
The question then arises, whether
it is sufficient to run this test 
for only one such pair,
that is, whether given 
a suitable circuit \calM\  for some
$\sigma$, $\rho$ and $n$, such that
$\stackrel{\sigma,\rho}{\leftrightarrow} $ refines
$\approx^{n+1}_{\tau,\pi}$,
one can for every sequence of pairs
$\sigma_i,\rho_i$, $i \ge 1$, for which
$\stackrel{\sigma_i,\rho_i}{\leftrightarrow} $ refines
$\approx^{n+i}_{\tau,\pi}$, 
obtain by pumping the components of \calM\
a suitable circuit where every subcircuit 
is closed under
$\crochet{ \star  \calJ^{\sigma_i,\rho_i} }  $.
\newline
We now develop tools that make it possible to explore pumping
in some depth.
Let $M \subset \BbbM_{A,B}$ be a set closed under
$\stackrel{\sigma,\rho}{\leftrightarrow} $
and let $Z \subset ports(M)$ be 
$\stackrel{\sigma,\rho}{\leftrightarrow} $-stable,
which means that pumping $M$ along $Z$ is possible. 
Let $\nu(x) = J$ for every $x \in Z$
and let $T$ be a set of tuples defined over $M$,
with for each $j \in J$ at least one element $t$
such that $\varphi(t) = j$.   
We work with a tuple 
$\mathfrak{t}$ be defined on
$\hat{\mathfrak{m}}  \in M^{(\theta,Z)} $
and built from $T$ through
\emph{consistent insertions}, where a tuple $t \in T$ can be
inserted at a port $x$ only if
$\varphi(t) = \psi(x)$. 
\newline
Within $\hat{\mathfrak{m}} $ we associate to
every copy $m$ of an element of $M$ 
the parameters depth and height, as follows.
The  
depth level $d=1$ in $ \hat{\mathfrak{m}} $ 
is a singleton $\bar{m}^1$ containing
a unique copy of a multicontext of $M$.
Then recursively, depth level $d+1$ is the set $ \bar{m}_{d+1} $ of all copies of
elements of $M$ inserted at the ports of $Z( \bar{m}_{d} )$,
which we call the $Z$-ports of $ \bar{m}_{d}$. 
Then $m$ is at depth $d$ iff $m \in  \bar{m}_{d}$. 
We define the height 
together with the  equivalence class
$ \bitruc{ \hat{\mathfrak{m}} }_{\sigma,\rho} =  \bitruc{ M^{(\theta,Z)} }_{\sigma,\rho} $
of 
$\hat{\mathfrak{m}}  \in M^{(\theta,Z)} $ under
$ \stackrel{\sigma,\rho}{\leftrightarrow} $,
as follows.
If $\theta < \sigma$,
then 
$ \bitruc{ \hat{\mathfrak{m}} }_{\sigma,\rho} =  M^{(\theta,Z)} $,
all $Z$-ports of $ \hat{\mathfrak{m}} $ 
are located at depth $\theta$, and 
the height of $m \in  \bar{m}_{d}   $ is  defined 
as the class
of $h=\theta-d$ under the congruence $\equiv_{\sigma,\rho}$.
%
Otherwise $\theta \ge \sigma$, 
and the class
$ \bitruc{ \hat{\mathfrak{m}}  }_{\sigma,\rho} $
consists in the union of all
sets
$M^{(\theta',Z)} $ for which
$\theta' \equiv_{\sigma,\rho} \theta$ plus,
because $ \stackrel{\sigma,\rho}{\leftrightarrow} $
is a congruence, 
every multicontext obtained  from
an element of $ \bitruc{ \hat{\mathfrak{m}}  }_{\sigma,\rho} $
by replacing a sub-multicontext 
$\hat{\mathfrak{n}}$
with some $\hat{\mathfrak{n}}' $
that satisfies
$\hat{\mathfrak{n}}'  \stackrel{\sigma,\rho}{\leftrightarrow}  \hat{\mathfrak{n}}$.
\newline
Let  $ \hat{\mathfrak{m}} '$ be obtained from 
$ \hat{\mathfrak{m}} $ by this substitution,
and let $x \in ports( \bar{m}_{d})$ be the port in which 
$\hat{\mathfrak{n}}$ has been replaced with $\hat{\mathfrak{n}}'$.
If $\hat{\mathfrak{n}} \in M^{(\delta,Z)}  $ for some $\delta < \sigma$,
then
$\hat{\mathfrak{n}}'  \in  M^{(\delta,Z)}  $,
and the height of $x$ and of the $m \in  \bar{m}_{d}$ it belongs to,
is unmodified.
Otherwise, $\hat{\mathfrak{n}} \in   \bitruc{ M^{(\delta,Z)}  }_{\sigma,\rho} $
with $\delta \ge \sigma$, and the height of $\bar{n}_{1}$, the top
depth level in $\hat{\mathfrak{n}}$, is 
the class
of $h(n)=\theta-1$ under the congruence $\equiv_{\sigma,\rho}$.
Making the induction hypothesis, that
$\hat{\mathfrak{n}}'  \stackrel{\sigma,\rho}{\leftrightarrow}  \hat{\mathfrak{n}}$
implies $h(n')\equiv_{\sigma,\rho}  h(n)$, we conclude that
the height of $x$ and of the $m \in  \bar{m}_{d}$ it belongs to,
is unmodified.
\newline
Let $x$  be located in  $ \hat{\mathfrak{m}} $  at  height $h$ and depth $d$:
the subforest
$\Delta(\hat{\mathfrak{m}}  ,x) $ belongs to $\bitruc{M^{(h,Z)} }_{\sigma,\rho}$
and
the context
$\nabla(\hat{\mathfrak{m}}  ,x) $ can be seen as the result of taking
a multicontext in $ \bitruc{M^{(d,Z)}}_{\sigma,\rho}$
consisting in the top $d$ depth levels of  $ \hat{\mathfrak{m}} $,
and inserting
an element of $\bitruc{M^{(h,Z)}}_{\sigma,\rho}$ at every $Z$-port of
$ \bar{m}_d $ other than $x$.
Next, given another multicontext
$\hat{\mathfrak{m}}' \in \bitruc{ M^{(\theta',Z)}}_{\sigma,\rho}$,
 $d' \le \theta'$ and $x' \in Z(\bar{m}_{d'}' ) $,
we have
\[
 x \stackrel{\sigma,\rho}{\leftrightarrow} x' 
\Leftrightarrow
\left(
\Delta(\hat{\mathfrak{m}}  ,x) \stackrel{\sigma,\rho}{\leftrightarrow} \Delta(\hat{\mathfrak{m}}'  ,x') 
\wedge
\nabla( \hat{\mathfrak{m}}  ,x) \stackrel{\sigma,\rho}{\leftrightarrow} \nabla(\hat{\mathfrak{m}}'  ,x') 
\right)
\Leftrightarrow
\left\{
\begin{array}{l}
d \equiv_{\sigma,\rho} d'
\\
h \equiv_{\sigma,\rho} h' \end{array}
\right. .
\]
We say that the $Z$-ports of $\hat{\mathfrak{m}} $ are at
the bottom and that the other ports are on the side.
Then, given $y \in ports(\bar{m}_d) $ and
$y'\in ports(\bar{m}'_{d'})$ located on the sides,
with $m(y)$ and $m'(y')$ the copies of elements of $M$
they belong to, we have
\[
 y \stackrel{\sigma,\rho}{\leftrightarrow} y' \
\Leftrightarrow\
\left\{
\begin{array}{l}
\nabla( m(y)  ,y) \stackrel{\sigma,\rho}{\leftrightarrow} \nabla( (m'(y') ,y') \\
d \equiv_{\sigma,\rho} d'
\\
h \equiv_{\sigma,\rho} h' 
\end{array}
\right. .
\]
%
%
%
Let 
$t = (\hat{\mathfrak{m}},\nu,\psi)$
and $t' = (\hat{\mathfrak{m}}',\nu,\psi)$ satisfy
$ {t} \, \crochet{ \star  \calJ^{\sigma,\rho} } \,  {t}'$,
so that
$\hat{\mathfrak{m}}  \stackrel{\sigma,\rho}{\leftrightarrow} \hat{\mathfrak{m}}'$,
and assume that  $\hat{\mathfrak{m}} \in   \bitruc{ M^{(\delta,Z)}  }_{\sigma,\rho} $
with $\delta \ge \sigma$.
Let $y \in ports(t)$ and $y' \in ports(t')$ be located on the sides,
at  depth levels  $d$ and $d'$, 
and heights $h$ and $h'$, respectively.
For every $u \in U^{\sigma,\rho}$,
they contribute to the counters  
$P[\calJ^{\sigma,\rho}](t)_{u,I,i}$ and 
$P[\calJ^{\sigma,\rho}](t')_{u,I,i}$, 
only if 
$ y \stackrel{\sigma,\rho}{\leftrightarrow} y' $.
Depending on $d$ we distinguish three cases:
if $d < \sigma$,
i.e. $y$ is located in the top $\sigma-1$ levels of $\hat{\mathfrak{m}}$,
this implies $d'=d$;
else if
$h < \sigma$,
i.e. $y$ is located at height less than $\sigma$ in $\hat{\mathfrak{m}}$,
this implies $h'  = h$.
Hence 
if $ {t} \, \crochet{ \star  \calJ^{\sigma,\rho} } \,  {t}'$,
then the  components of
$P[\calJ^{\sigma,\rho}](t)$ and 
$P[\calJ^{\sigma,\rho}](t')$ corresponding to the top and bottom
regions match
on a level-by-level basis.
Otherwise $y$ is located
in the ``middle region'' of $\hat{\mathfrak{m}}$, where 
$d \ge \sigma  $ and $h \ge \sigma$, and
where
the ports are partitioned according to $(d - \sigma)\, \mathsf{mod}\, \rho$
and $(h - \sigma)\, \mathsf{mod}\, \rho$,
and the ports  within the same class contribute to the same counters.
The contributions of all ports within the same class are added up.
\newline
%
%
Until now, dealing with modular quantifier and arithmetics 
was no more complex than working only on  the variety
${\mathbf{**}}\bicrochet{\N_{\tau,1} } = \mathsf{FO}[\Panc]  $,
i.e. the case where $\pi=1$. The situation changes here,
as there will be places where
modular arithmetics adds an extra layer of complexity to the
work being done.
From now on, therefore, we restrict ourselves to the case where 
$\pi=1$.
\\

%
%
We now look at a type of circuit
where the multicontexts of $T_J$ are obtained by
pumping at the $Z$-ports of a set $M$,
i.e. $M_J \subset  \bitruc{M^{(\theta,Z)}}_{\sigma,\rho} $
for some $\theta \in \N$
and the multicontexts of $T_J$  have $Z$-ports
that are available for further pumping.
The sets of tuples
with this property is recursively enumerable.
We show that
for every $\chi > \sigma$
a set $U_J$ closed under
$\crochet{ \star  \calJ^{\chi,1} }$
can be built from $T_J$ by pumping along $Z$.
If every subcircuit
of a circuit \calT\ that satisfies 
$\mathbf{RC}(\calG  {\star}  \calJ^{\sigma,1} )$
falls into this case, then 
for every $\chi > \sigma$
one can obtain by pumping \calT\
a circuit \calU\  that satisfies 
$\mathbf{RC}(\calG  {\star}  \calJ^{\chi,1} )$;
in other words, a uniform
proof-by-pumping for
$\calG \not\in {\mathbf{**}}\bicrochet{\N_{\tau,1} } $
can always be built from \calT.
Later in this section, we will 
look at another type of circuits
and show that in some circumstances,
$T_J$ cannot be used to build an entire proof.
\newline
Let the vector  $P[\calJ^{\sigma,\rho}](t)_Z$ 
consist of the components of
$P[\calJ^{\sigma,\rho}](t)$ associated to the
$Z$-ports, and
$P[\calJ^{\sigma,\rho}](t)_Y$
consist in all the other components.
With the notations $\bar{m}_1 = \{m^{\mathsf{top}}\}$ and
$\bar{m}'_{1} = \{m^{\mathsf{top}}{}'\}$, 
we define  tuples
$t^{\mathsf{top}} = (m^{\mathsf{top}},\nu,\psi)$
and $t^{\mathsf{top}}{}' = (m^{\mathsf{top}}{}',\nu,\psi)$
from $t$ and $t'$;
we observe that if
$t \, \crochet{ \star  \calJ^{\sigma,\rho} } \, t'$,
then
$
P[\calJ^{\sigma,\rho}](t^{\mathsf{top}})_Y  \equiv_{\tau,\pi} 
P[\calJ^{\sigma,\rho}](t^{\mathsf{top}}{}')_Y
$.
Also, since
$m^{\mathsf{top}}  \stackrel{\sigma,\rho}{\leftrightarrow}  m^{\mathsf{top}}{}' $,
we have
$|Z(m^{\mathsf{top}})| \equiv_{\sigma,\rho}   |Z(m^{\mathsf{top}}{}' )|$.
%
Given a
subcircuit $T_J = \{ t_{j} : j \in J \}$  of a circuit that satisfies
$\mathbf{RC}(\calG  {\star}  \calJ^{\sigma,\rho} )$,
we define in the same way 
the sets $M^{\mathsf{top}} = \{ m^{\mathsf{top}}_j : j \in J \}$
of multicontexts and $T^{\mathsf{top}} = \{ t^{\mathsf{top}}_j : j \in J \}$
of tuples,
and for every $\theta \ge 2$ we build by
consistent pumping 
the set
$T^{\mathsf{top}}{}^{(\theta,Z)} = \{ t^{(\theta,Z)}_j : j \in J \}$.
We associate to every $Z$-port $x$ of $M_J$ its depth
level $\theta_x$; it satisfies $\theta_x \ge \sigma$.
Then we build $U_J$
in three steps, where the first two are:
\begin{latindense}
\item build the set $T^{\mathsf{top}}{}^{(\chi,Z)}$; 
\item for each $j \in J$, build $\tilde{u}_j$ from $t^{(\chi,Z)}_j$,
by inserting
at every $Z$-port $x$ a copy of $t_{\psi(x)}$.
\end{latindense}
Every
$Z$-port $x$ of  $\tilde{u}_j$  is a copy of some $Z$-port $z$ of $T_J$
and is located at depth $\theta_x = \chi + \theta_z$.
For all $j,j' \in J$, the properties of $T^{\mathsf{top}}$ imply
\[
P[\calJ^{\chi,1}](t^{(\chi,Z)}_j)_Y =
P[\calJ^{\sigma,1}](t^{(\chi,Z)}_j)_Y  \equiv_{\tau,1} 
P[\calJ^{\sigma,1}](t^{(\chi,Z)}_{j'})_Y =
P[\calJ^{\chi,1}](t^{(\chi,Z)}_{j'})_Y
\]
which means that the counters corresponding to
the top  $\chi$ levels of $\tilde{u}_j$ and $\tilde{u}_{j'}$ match
(under the $\equiv_{\tau,1}$).
The third step is:
\begin{latindense}
\setcounter{latinzahl}{2}
\item build $u_j$ by inserting, at every $Z$-port $x$ of $\tilde{u}_j$,
a copy of $t_{\psi(x)}^{(\eta(x),Z)}$. 
\end{latindense}
To determine the integer $\eta(x)$, we
associate to $T^{\mathsf{top}}$ a $|J| \times  |J|$
matrix $A$ with entries in $\N_{\tau,1}$, where 
$A_{ij} = | \{ x \in Z(t^{\mathsf{top}}_i) : \psi(x) = j \}|$.
Then $(A^\delta)_{ij}$ is the number of $Z$-ports
$x$ in $t_{i}^{(\delta,Z)} $ such that $\psi(x) = j$.
The powers of this matrix constitute a finite semigroup 
with idempotent element $A^\omega$: if every $\eta(x)$ is a multiple
of $\omega$, then 
\[
P[\calJ^{\chi,1}](t_{\psi(x)}^{(\eta(x),Z)})_Z  \equiv_{\tau,1} 
P[\calJ^{\chi,1}](t_{\psi(x')}^{(\eta(x'),Z)})_Z 
\]
will hold for every pair $x,x'$ of $Z$-ports of  $\tilde{u}_j$
that satisfy $\psi(x)= \psi(x')$.
Since for all $j,j' \in J$ and for all  $t,t' \in T_J$,
\begin{equation} \label{eqn:PZY}
|Z(t^{(\chi,Z)}_j)|  \equiv_{\tau,1}   |Z(t^{(\chi,Z)}_{j'})|
\quad
\text{and} \quad
P[\calJ^{\sigma,1}](t)_Z  \equiv_{\tau,1} 
P[\calJ^{\sigma,1}](t')_Z 
\end{equation}
we have
$P[\calJ^{\chi,1}](\tilde{u}_j)_Z  \equiv_{\tau,1} 
P[\calJ^{\chi,1}](\tilde{u}_{j'})_Z $
and from there
$P[\calJ^{\chi,1}](u_j)_Z  \equiv_{\tau,1} 
P[\calJ^{\chi,1}](u_{j'})_Z $.
We now claim that if every $\eta(x)$ satisfies,
$\eta(x) \ge \tau + \chi$, then
$P[\calJ^{\chi,1}](u_j)_Y  \equiv_{\tau,1} 
P[\calJ^{\chi,1}](u_{j'})_Y $.
This already holds for the top $\chi$ levels in 
${u}_j$ and ${u}_{j'}$.
Meanwhile,  the lower $\chi$ levels in
$t_{\psi(x)}^{(\eta(x),Z)}$ consist of copies of
elements of $T^{\mathsf{top}}$;
their numbers in ${u}_j$ and ${u}_{j'}$
match, thanks to Equation (\ref{eqn:PZY});
this ensures a match in the bottom region.
The middle region consists in  the tuples inserted at step \emph{ii}
and in elements of
$T^{\mathsf{top}}{}^{(\xi(x) - \chi,Z)}$, namely the
tuples inserted at step \emph{iii} minus their
lowermost $\chi$ levels.
By Equation  (\ref{eqn:PZY}) and the fact that
$\eta(x) - \chi \ge \tau$ for every $Z$-port $x$ of ${u}_j$ and ${u}_{j'}$,
the contributions of the tuples inserted at step \emph{iii}
match with one another.
Concerning the tuples inserted at step \emph{ii},
let   $t,t' \in T_J$ and consider 
two non-$Z$-ports $y,y'$,
located at depth levels $d,d'$ and height levels $h,h'$ 
in $t$ and $t'$, respectively, denoting by 
$m(y)$ and $m'(y')$ the copies of multicontexts from $M$
they belong to.
If
$y \stackrel{\sigma,1}{\leftrightarrow} y'$,
then
$\nabla(m(y),y)    \stackrel{\sigma,1}{\leftrightarrow} \nabla(m'(y'),y') $.
The reverse implication works in $t,t'$ only if 
$d \equiv_{\sigma,1} d'$ and
$h \equiv_{\sigma,1} h'$.
Let $t$ be inserted at step \emph{ii} of the construction of  $\tilde{u}_j$:
the same port $y$
is, in  $\tilde{u}_j$, at depth $d+\chi \ge \chi$ and height at least $\tau + \chi+h \ge \chi$;
then 
if $t'$ is in the same $\tilde{u}_j$, we have
\[
\nabla(m(y),y)    \stackrel{\sigma,1}{\leftrightarrow} \nabla(m'(y'),y') 
\ \Rightarrow\ y \stackrel{\chi,1}{\leftrightarrow} y'.
\]
Then from this, from
$P[\calJ^{\sigma,1}](t)_Y  \equiv_{\tau,1} 
P[\calJ^{\sigma,1}](t')_Y $, and from
the fact that the numbers of occurrences in
$\tilde{u}_j$ and $\tilde{u}_{j'}$ of elements of $T_J$ 
match because
$|Z(t^{(\chi,Z)}_j)|  \equiv_{\tau,1}   |Z(t^{(\chi,Z)}_{j'})|$,
we conclude that the counters corresponding to
the middle regions of $u_j$ and $u_{j'}$ balance.
%
%
%
%
\\

%
We now look at a circuit where $T_J$ has been built
by pumping along $Z$ but has no
$Z$-ports, and the
construction of $T_J$  uses two sets  of multicontexts $M$ and $R$ closed under
 $  \stackrel{\sigma,1}{\leftrightarrow} $,
 and is done in two steps:
pumping at the $Z$-ports of a set $M$
to obtain a set $S_J = \{ s_j : j \in J \}$,
then for each $j \in J$,
inserting at every  $Z$-port $x$ of $s_j$
a tuple $r(x) \in R$, which builds $t_j$.
This is a situation where
$P[\calJ^{\sigma,1}](s)  \equiv_{\tau,1} 
P[\calJ^{\sigma,1}](s') $
is not satisfied by some tuples $s,s' \in S_J$
(otherwise a suitable circuit could be built
solely with $S_J$,
without the tuples from $R$), but where
$P[\calJ^{\sigma,1}](t)  \equiv_{\tau,1} 
P[\calJ^{\sigma,1}](t') $
holds for all $t,t' \in T_J$.
From the latter equivalence we deduce that
$P[\calJ^{\sigma,1}](s)_Y  \equiv_{\tau,1} 
P[\calJ^{\sigma,1}](s')_Y $
holds for all  $s,s'$, and we have
$P[\calJ^{\sigma,1}](s)_Z  \not\equiv_{\tau,1} 
P[\calJ^{\sigma,1}](s')_Z $ for some $s,s'$.
This makes possible a situation where a different tuple
from $R$ has to be inserted at every port of the $s_j$'s
in order for the resulting $t_j$'s to be equivalent under
$\crochet{ \star  \calJ^{\sigma,1} }$: this means that
tuples from $R$ inserted at
the antichains of ports $Z(s_j)$ can be correlated
in a potentially intricate way.
\newline
We assume that using $R$  was necessary, namely
that no set of tuples $S = \{ s_j : j \in J \}$
built by
pumping at the $Z$-ports of any set $M$
closed under
 $  \stackrel{\sigma,1}{\leftrightarrow} $
can simultaneously be diagonal
and closed under $\crochet{ \star  \calJ^{\sigma,1} }$.
Equivalently,  for every vector of  counters $\vec{p}$, there is a
subset $J' \subset J$ such that, for every $s \in S$,
having
$P[\calJ^{\sigma,1}](s)  \equiv_{\tau,1} \vec{p}$
implies $\varphi(s) \not\in J'$.
Moreover, the set $R$ of tuples to be inserted at the
ports of $S$ does not satisfy this either, for otherwise
some subset $\{ r_j : j \in J \}$
would be suitable to replace $T_J$ in the construction of a circuit
that satisfies
$\mathbf{RC}(\calG  {\star}  \calJ^{\sigma,1} )$.
%
%
Provided that $M$ in not itself expressible\footnote{This is 
the counterpart, in the world of trees and multicontexts,
of the fact that the word language
$\{ w^\theta : \theta \ge \sigma \}$ is star-free,
unless $w = x^\pi$ for some word $x$
and $\pi \ge 2$.}
 as
$N^{(\pi,Z)}$ for $N \subset \BbbM_{A,B}$ and $\pi \ge 2$, 
one can prove by induction on the construction of
$\bitruc{ M^{(\sigma,Z)} }_{\sigma,1}$
that this set 
is recognized
by a forest algebra 
in ${\mathbf{**}}\bicrochet{\N_{\tau,1}}$,
and that the same holds for every intersection\footnote{Actually,
one can  prove that
$\bitruc{ T^{(\sigma,Z)} }_{\sigma,1}$
is a union of a number of such classes.}
 of the corresponding set of tuples
$\bitruc{ T^{(\sigma,Z)} }_{\sigma,1}$
with an equivalence class for
$\crochet{  \star  \calJ^{\sigma,1} }$;
let
$\calH_1 \in {\mathbf{**}}\bicrochet{\N_{\tau,1}}$
recognize everyone of these intersections.
Next, the constraint on $R$ implies that,
for a large enough $n$, no set of witnesses for $J$
and $\approx^n_{\tau,1}$
can be built from $R$.
Since
 $  \stackrel{\sigma,1}{\leftrightarrow} $
refines $\approx^n_{\tau,1}$,
this means that
no  equivalence class for
$\crochet{  \star  \calJ^{\sigma,1} }$
contains a diagonal subset of $R$.
Then in
the very special case where each intersection
of such a class with $R$ 
contains elements from at most one set $\varphi^{-1}(j)$, $j \in J$,
and given the fact that $R$ is a set of tuples built
over a set of multicontexts closed under
 $  \stackrel{\sigma,1}{\leftrightarrow} $,
there exists an algebra 
$\calH_2 \in {\mathbf{**}}\bicrochet{\N_{\tau,1}}$
that, for every $t \in T_J$, determines
the content of $ P[\calJ^{\sigma,1}](s)_Z $
and the class of $s$ for $\crochet{  \star  \calJ^{\sigma,1} }$,
and therefore the subset $J' \subset J$ for which
we have $\varphi(t) \not\in J'$. 
This means that for every
$n' \in \N$ such that both
$\calH_1  \prec \calH^{n'}_{\tau,1}$
and
$\calH_2   \prec \calH^{n'}_{\tau,1} $ hold,
no set of tuples
obtained by pumping $S$ and inserting
elements of $R$ at the $Z$-ports
can constitute a subcircuit for $J$
in a circuit that satisfies
$\mathbf{RC}(\calG  {\star}  \calH^{n'}_{\tau,1}  )$.
%
%
This special case is a situation where
the existence of a
circuit that satisfies
$\mathbf{RC}(\calG  {\star}  \calJ^{\sigma,1}  )$,
where
$\stackrel{\sigma,1}{\leftrightarrow} $ refines
$\approx^{n+1}_{\tau,1}$,
does not imply the existence of
a uniform proof-by-pumping
for
$\calG \not\in  {\mathbf{**}}\bicrochet{\N_{\tau,1}}$.
%



%



\section{Examples of  proofs}  \label{sec:exaproof}

We show how the techniques and notations of this article 
work on examples, of which three are
discussed in the literature under different formalisms
and methods.

\subsection{The Boolean algebra}  \label{sec:boole}

The Ehrenfeucht-Fra\"{\i}ss\'e game concerning this algebra, described in
\cite[Theorem 4.2]{po94}, is an example
of a proof-by-copy.
\newline
In the
Boolean forest algebra \calB, the horizontal monoid is
the direct product $B=\{\sfzero,\sfone\} \times \{\sfzero,\sfone\} $
of the AND and OR monoids (in the first and second component, respectively),
so that $\crochet{\sfone,\sfzero} $ is the identity and 
$\crochet{\sfzero,\sfone}$ is absorbing. Besides the elements of the
form $\varepsilon + \crochet{\mathsf{a},\mathsf{b}} $,
the vertical monoid is generated by
$\varphi(\wedge)$ and $\varphi(\vee)$, which work on each $\crochet{\mathsf{a},\mathsf{b}} \in B$ 
as projections on
the first and second component, respectively.
With $A = \{\wedge,\vee\}$, the Boolean algebra can be regarded as the
image of \Sfree\ by the homomorphism $\varphi$; in the special case of
one-node forests, this gives 
\[
\varphi(\wedge) = \varphi(\wedge ) \crochet{\sfone,\sfzero} 
= \crochet{\sfone,\sfone} 
\quad \text{and} \quad
\varphi(\vee)  = \varphi(\vee ) \crochet{\sfone,\sfzero}  = \crochet{\sfzero,\sfzero} .
\]
We 
denote by $L_\sfzero$ and $L_\sfone$  the languages
$\varphi^{-1}(\crochet{\sfzero,\sfzero} )$ and
$\varphi^{-1}(\crochet{\sfone,\sfone} )$, respectively.
The proof described in Potthoff \cite[Theorem 4.2]{po94}\
builds witnesses from a unique circuit that consists in
two tuples $t_\sfzero = (m,\nu_\sfzero,\psi_\sfzero)$ and $t_\sfone = (m,\nu_\sfone,\psi_\sfone)$
defined over the strongly connected set 
$J =  \{\crochet{\sfzero,\sfzero},  \crochet{\sfone,\sfone}\}$, 
where the mappings $\nu_0$ and $\nu_\sfone$ are constant
(map every port to $J$) while $\psi_\sfzero$ and $\psi_\sfone$ 
are
depicted on Figure 1; this circuit satisfies
$ \mathbf{RC}(\calB {\star} \calH^n_{\tau,\pi} )$
for every combination of $\tau$ and $\pi$.

\begin{figure}
$$\diagram
  & &    \wedge \dlto \drto   &   &  &  &  &  \wedge \dlto \drto   &  &  &  \\
& \vee \dlto \dto    &    &   \vee \dto \drto  & & & \vee \dlto \dto  &   &   \vee \dto \drto  & \\
 \sfzero &     \sfzero  &   &  \sfone  & \sfone &  \sfzero &     \sfone  &   &    \sfzero   & \sfone \\
\enddiagram$$
\caption{ tuples $t_\sfzero$ (left) and $t_\sfone$ (right) }
\end{figure}

\subsection{The duplex Boolean formulas}  \label{sec:duplex}

The ``duplex'' technique is a way of building examples of algebras
that demand proofs where the
witnesses have to be built simultaneously for two  $\approx^\N$-tight
sets. We describe the duplex Boolean formulas as an example.
The language consists of those
trees that satisfy  the following conditions.
\begin{listedense}
\item There is at least one interior node.
\item
If a node $x$ is a leaf, then its label
$\lambda(x)$ belongs to $\{\sfzero,\sfone\}$;
otherwise 
$\lambda(x) \in \{ \cap,\cup,\sqcap,\sqcup \}$.
\item
The sons of a node with  $\lambda(x) \in \{ \cap,\sqcap \}$
are either two leaves, or four interior nodes $y,y',z,z'$
with $\lambda(y) = \lambda(y') = \cup$ and $\lambda(z) = \lambda(z') =   \sqcup$.
\item
The sons of a node with  $\lambda(x) \in \{ \cup,\sqcup \}$
are either two leaves, or four interior nodes $y,y',z,z'$
with $\lambda(y) = \lambda(y') = \cap$ and $\lambda(z) = \lambda(z') =   \sqcap$.
\item
Each node $x$ has a value $val(x) \in \{ \sfzero,\sfone,\bot\}$,
defined according to the following rules.
\begin{listedense}
\item A leaf evaluates to its label, e.g. if $\lambda(x) = \sfone$, then 
$val(x) = \sfone$.
\item  An interior node $x$
with two sons $w,w'$, which are leaves, evaluates to 
$val(w) \wedge val(w')$ if  $\lambda(x) \in\{ \cap,\sqcap \}$,
and to $val(w) \vee val(w')$ if $\lambda(x) \in\{ \cup,\sqcup \}$.
\item If $x$ has four sons
and
$\lambda(x) \in \{ \cap,\sqcap \}$, then $val(x) \neq \bot$
iff none of its sons evaluates to $\bot$ 
and $ val(y) \wedge val(y') = val(z) \wedge val(z')$,
in which case
$val(x) = val(y) \wedge val(y') $;
in other words, the AND of the
subforests with roots in 
$\{ \cap,\cup\}$
and $\{ \sqcap,\sqcup\}$ are evaluated separately, 
and $val(x) \neq \bot$ iff the values coincide.
When  $x$ has four sons
and
$\lambda(x) \in \{ \cup,\sqcup \}$,
$val(x)$ is an OR, obtained in a similar way.
\end{listedense}
\end{listedense}
\begin{figure}
$$\diagram
 & \cap \dlto \drto        & &    & \cap \dlto \drto  & & & \sqcap \dlto \drto        & &    & \sqcap \dlto \drto  & \\
 \sfzero &   &  \sfzero  &       \sfone  & & \sfone & \sfzero &   &  \sfzero  &       \sfone  & & \sfone   \\
\enddiagram$$
\caption{ from left to right, the trees $ t(\circleddash,\sfzero)$, $t(\circleddash,\sfone) $, $ t(\boxdot,\sfzero)$ and $t(\boxdot,\sfone) $}
\end{figure}
The language of those trees where
$\lambda(root) = \cap$ and $val(root) = \sfone$
is not first-order definable.
Witnesses for this are built from an array of forests \calT\
consisting in the trees 
$ t(\circleddash,\sfzero)$, $t(\circleddash,\sfone) $, $ t(\boxdot,\sfzero)$ and $t(\boxdot,\sfone) $,
and from the circuit \calM, consisting in
$M( \circleddash) = \{ m(\circleddash,\sfzero),\, m(\circleddash,\sfone) \}$ and
$M(\boxdot) = \{ m(\boxdot,\sfzero),\, m(\boxdot,\sfone) \}$,
over the alphabets 
$A = \{ \sfzero,\sfone,\cap,\cup,\sqcap,\sqcup \}$
and 
$B = \{ \circleddash,\boxdot \}$.
Each multicontext consists in one root, with label $\cap$ in 
those of $M(\circleddash)$ and $\sqcap$ in $M(\boxdot)$,
then
four nodes $y,y',z,z'$
with $\lambda(y) = \lambda(y') = \cup$ and $\lambda(z) = \lambda(z') =   \sqcup$;
under each of them are four ports.
In $m(\circleddash,\sfzero)$ and $m(\boxdot,\sfzero)$, two of
the ports below node $y$ carry the label $(\circleddash,\sfzero)$
and the other two the label $(\boxdot,\sfzero)$;
the same holds for the ports located under $z$.
Under the node $y'$,
two ports  carry the label $(\circleddash,\sfone)$
and the other two  $(\boxdot,\sfone)$;
the same holds for the ports located under $z'$.
Thus, under nodes $y$ and $y'$, the 
\sfzero's and \sfone's of the four port labels in
$ \{ \circleddash\} \times \{ \sfzero,\sfone \} $
follow the same pattern as the four ports of 
the $m_\sfzero$ multicontext defined in
the previous section.
The same observation holds for the ports labels in
$ \{ \boxdot\} \times \{ \sfzero,\sfone \} $,
as well as the ports under $z$ and $z'$.
Similarly, the port labels in
$m(\circleddash,\sfone)$ and $m(\boxdot,\sfone)$
 follow the pattern of the ports of $m_\sfone$.
\newline
Each array of witnesses $\calS^{(n)} = \calM^n {\cdot} \calT$ 
consists in
two subarrays $S^{(n)}(\circleddash)$ and $S^{(n)}(\boxdot)$,
whose elements 
can be told apart by  looking at the label of their roots.

\subsection{Forest algebras with  aperiodic vertical monoid
and uniform vertical confusion}  \label{sec:evendepth}

We describe two examples of forest algebras where the
vertical monoid is aperiodic
and that have
vertical confusion on uniform multicontexts.
\newline
The first example,  discussed in \cite{po95},
is the set $L_{\text{even}}$  of those binary trees over a trivial alphabet $A = \{{\bullet}\}$
where all leaves are located at a nonzero even depth.
In its syntactic forest algebra, the horizontal monoid is
$ \{ 0, \Bfe,\Bfo,\Bfe\Bfe,\Bfo\Bfo,\infty\}$,
where  $0$ is the identity,
$\Bfe + \Bfe =  \Bfe\Bfe$, $\Bfo + \Bfo = \Bfo\Bfo$,
and everything else evaluates to the absorbing element $\infty$.
The one-port context ${\bullet} \square$ is mapped to the
vertical monoid element that maps 
$0$ and $\Bfe\Bfe$ to $\Bfo$,
$\Bfo\Bfo$ to $\Bfe$, 
and everything else  to $\infty$.
The  vertical monoid of its syntactic forest algebra\footnote{In
this example and the next ones,  the action on $0$ of the
generators of the vertical monoid is not defined; the
size varies depending on how  this gap 
is filled. The other properties that matter here are not affected.
}
is aperiodic, 
has $19$ elements; the non-idempotent ${\bullet} \square$
constitutes a \calJ-class by itself; the others constitute three regular
\calJ-classes analogous to those of the Brandt monoid $B_4$.
Vertical confusion is verified with the multicontext with one interior node,
that is the father of two ports, i.e. ${\bullet} (\square + \square)$.
\newline
The Ehrenfeucht-Fra\"{\i}ss\'e game that proves that
$L_{\text{even}}$ is not in $\mathsf{FO}[\Panc]  $ uses
nontrivial uniform multicontexts
whose interior nodes constitute the set
$\{ m_d : d \ge 1 \}$, where $m_d$ is the complete binary tree
of depth $d$ over $A$, and where each node at depth $d$ is the father of two
ports. Given a labeling $\lambda : ports(m_d) \rightarrow \{ \Bfe,\Bfo,\Bfe\Bfe,\Bfo\Bfo \}$,
we have $\Brvphi(m_d,\lambda) \neq \infty$ only if $\lambda$ is a constant mapping
with image in $\{ \Bfe,\Bfo \}$.
A suitable set $\calS^{1}$ consists in 
$s_{\Bfe} $ and $s_{\Bfo} $, where
$s_{\Bfe} $ and $s_{\Bfo} $ are built from $m_{2\tau}$ 
and $m_{2\tau+1}$, respectively, by inserting $\Bfe$,
at every port.
Let $\sigma(n) = 2\tau(2^{n+1}+1)$.
The circuit
$\calM^{(n+1)} $ consists in multicontexts  $m_{\sigma(n)} $.
and
$m_{\sigma(n)+1} $, where all ports carry the same label,
satisfies
$ \mathbf{RC}(\calG {\star} \calH^n_{\tau,\pi} )$.
and therefore can be used 
to build $\calS^{(n+1)} $.
\\

%
Our second example is a language defined as follows:
\begin{listedense}
\item[\textbf{(1)}] it consists of trees which are either one-node with label $\sfzero$, $\sfone$, \sftwo\ or $\bot$,
or built recursively through insertions in the four-port uniform multicontext $m = a( b( y_0 + y_1) + b( z_0 + z_1))$;
node labels are taken in
$\{a,b\}$
and
port labels are undefined;
\item[\textbf{(2)}] among them, exactly those which evaluate to \sfzero\ according to the two tables below
belong to the language. 
\end{listedense}

\begin{center}
\begin{tabular}{c||c|c|c|c|cc||c|c|c|c|}
$\ a \ $ & $\sfzero$ & $\sfone$ & $\sftwo$ & $\bot$ &
 \hspace{0.5in} &  $\ b \ $ & $\sfzero$ & $\sfone$ & $\sftwo$ & $\bot$ \\
\cline{1-5} \cline{7-11} 
\cline{1-5} \cline{7-11} 
$\sfzero $ & $\sfzero$ & $\bot$ & $\bot$ & $\bot$ & 
\hspace{0.5in} &  $\sfzero \ $ & $\sfone$ & $\bot$ &  $\bot$ & $\bot$ \\
\cline{1-5} \cline{7-11} 
$\sfone $ & $\bot$ & $\sftwo$ & $\bot$ & $\bot$ & 
\hspace{0.5in} &  $\sfone \ $ & $\bot$ & $\bot$ &  $\bot$ & $\bot$ \\
\cline{1-5} \cline{7-11}  
$\sftwo $ & $\bot$ & $\bot$ & $\bot$ &  $\bot$ &
 \hspace{0.5in} &  $\sftwo \ $ & $\bot$ & $\bot$ &  $\sfzero$ & $\bot$ \\
\cline{1-5} \cline{7-11} 
$\bot $ & $\bot$ & $\bot$ &  $\bot$ & $\bot$ & 
\hspace{0.5in} &  $\bot \ $ & $\bot$ & $\bot$ &  $\bot$ & $\bot$ \\
\cline{1-5} \cline{7-11} 
\end{tabular}
\end{center}
Besides the identity $0$ and the absorbing element $\infty$,
the horizontal monoid has one element for each of the following:
a tree that evaluates to $\sfzero$, to \sfone\ and to \sftwo,
and a sum of two such trees.
The  vertical monoid of its syntactic forest algebra contains three regular
\calJ-classes, analogous to those of the Brandt monoid $B_6$.
Vertical confusion is verified with the multicontext $m$ given above.


\subsection{The Potthoff Algebra}  \label{sec:potthoff}

This is the syntactic
algebra of a  language of trees over
$A = \{ \sfzero, \sfone, \bot, \vartriangle,\doteqdot \}$ defined  as follows  (see \cite[Definition 25]{po95}):
\begin{listedense}
\item[\textbf{(1)}] it is included in a set which consists in
three one-node trees with labels $\sfzero$, $\sfone$ or $\bot$,
and of all those trees 
built recursively through insertions in the three-port multicontext $m = \vartriangle( z + \doteqdot( y_0 + y_1))$,
depicted on Figure 3, where the
port are labelled with variable names;
\item[\textbf{(2)}] it consists of exactly those 
trees which evaluate to \sfone\ according to the following two tables.
\end{listedense}

\begin{center}
\begin{tabular}{c||c|c|c|cc||c|c|c|}
$\ \vartriangle \ $ & $\sfzero$ & $\sfone$ & $\bot$ & \hspace{0.5in} &  $\ \doteqdot \ $ & $\sfzero$ & $\sfone$ & $\bot$ \\
\cline{1-4} \cline{6-9} 
$\sfzero $ & $\sfone$ & $\bot$ & $\bot$ & \hspace{0.5in} &  $\sfzero \ $ & $\sfone$ & $\sfzero$ & $\bot$ \\
\cline{1-4} \cline{6-9} 
$\sfone $ & $\bot$ & $\sfzero$ & $\bot$ & \hspace{0.5in} &  $\sfone \ $ & $\sfzero$ & $\sfone$ & $\bot$ \\
\cline{1-4} \cline{6-9} 
$\bot $ & $\bot$ & $\bot$ & $\bot$ & \hspace{0.5in} &  $\bot \ $ & $\bot$ & $\bot$ & $\bot$ \\
\cline{1-4} \cline{6-9} 
\end{tabular}
\end{center}

Condition \textbf{(1)} can be specified with a first-order formula in $\mathsf{FO}[\Panc]$;
Condition \textbf{(2)} is specified with a formula in  $\mathsf{FOMod_2}[\Panc]$
which, at every `$\vartriangle$' node $x$, consider three of the paths that 
start at $x$ and end at a leaf, namely
the  unique path that traverses only `$\vartriangle$' nodes,
through copies of the port $z$, and that we call the \aaa-path of $x$,
plus the \aaa-path of each  `$\vartriangle$' node that is a grandson of $x$
through  its  `$\doteqdot$'  son. The formula
verifies the parities of their lengths and the label of the leaves
located at their end.


\begin{figure}
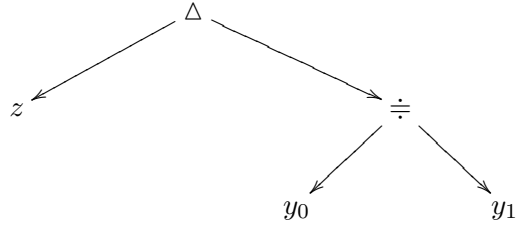

$$\diagram
 & &    \vartriangle \dllto \drrto   &   &  &  \\
z  &    & &    & \doteqdot \dlto \drto  & \\
   &    &     &  y_0  & &  y_1   \\
\enddiagram$$
\caption{the multicontext $m$}
\end{figure}

The reader can verify that the horizontal monoid $G$ of this algebra $\calG = (G,W)$,
consists of $10$ elements: 
an identity, an absorbing element, and eight mutually accessible elements, two of which
can be interpreted as `a $\vartriangle$ node with output \sfzero' and  `a $\vartriangle$ node with output \sfone';
we denote these elements \hatzero\ and \hatone, respectively.
Using Section \ref{sec:eqpump}\ it can be verified that
$\leftrightarrow_{3,1} $ refines the canonical congruence of \calG.

\subsubsection{Verifying non-definability with the extended algebra}

One thing that places the Potthoff algebra  outside of
$ {\mathbf{**}}\crochet{\N_{\tau,1} }$
for any $\tau \ge 1$, is that the vertical monoid $W_\%$ of
its extended algebra $\calG_\%$
is non-aperiodic. 
To see this, we define $B = \{b,b'\}$, 
$C_{b} = \{ c_0 \}$ and  $C_{b'} = \{ c_0' , c_1' \}$;
the extension of $\varphi$ to $C$ is given by
$\varphi(c_0) = \varphi(c_0') = \sfzero$
and
$\varphi(c_1') = \sfone$.
Consider a labeling where
$\nu(z) = b$
and
$\nu(y_0) = \nu(y_1) = b'$
and we do not care about $\psi$.
We define a sequence of tuples $t_i$, $i \ge 1$,
where $t_1 = (m,\nu,\psi)$ and every
$t_i$ is obtained by inserting $t_{i-1}$ at the
$z$ port of $t_1$.
The corresponding elements of $W_\%$
satisfy
$\varphi_\% (\nabla(t_i,z)) \neq \varphi_\% (\nabla(t_{i+1},z)) $
and
$\varphi_\% (\nabla(t_i,z)) = \varphi_\% (\nabla(t_{i+2},z)) $
for every $i \ge 1$; in particular,
$\varphi_\% (\nabla(t_i,z)) $
maps $\{\sfzero\}$ to $\{ \sfzero,\bot \}$ if $i$ is even
and to $\{ \sfone,\bot \}$  if $i$ is odd.
In other words, the group $\Z_2$ divides the vertical monoid
of $\calG_\%$.
The labeling used here,
where $\nu(y_0) = \nu(y_1) = b'$ means that the images by $\varphi$
of the lateral subtrees are not taken into account,
enables $\calG_\%$ to test explicitly the
parity of the length of the path between the $z$ port and
the root, i.e. the \aaa-path of $t_i$. 

\begin{figure}
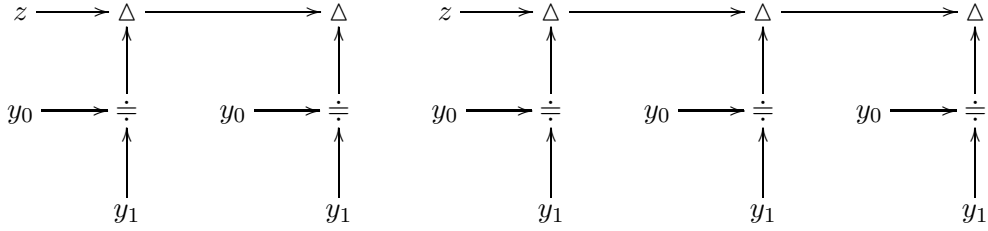

$$\diagram
z \rto  &    \vartriangle   \rrto &   &    \vartriangle   &
 z \rto  &    \vartriangle   \rrto &   &    \vartriangle   \rrto &   &    \vartriangle   \\
 y_0  \rto & \doteqdot  \uto  & y_0 \rto  & \doteqdot  \uto  & 
  y_0  \rto & \doteqdot  \uto  & y_0  \rto & \doteqdot  \uto   & y_0 \rto  & \doteqdot  \uto   \\
    &   {y_1} \uto &  &   y_1\uto  & &   {y_1} \uto &  &   y_1 \uto  &  &   y_1 \uto &   \\
\enddiagram$$
\caption{Tuples $t_2$ (left) and $t_3$  (right).}
\end{figure}

\subsubsection{Verifying non-definability with a recursive proof}

%
Potthoff wrote his proof that $\calG$ is not first-order
in his  doctoral thesis \cite{po94b}\ but apparently did not publish it elsewhere.
The  proof described here is similar in its main lines;
the differences come from the fact that  Potthoff's construction
builds witnesses for the congruence defined by the first-order formulas
of quantifier depth $n$, whereas we are interested in witnesses for
congruences of the form $\approx^n_{\tau,1}$.
\newline
We recall that
$A = \{ \sfzero, \sfone, \bot, \vartriangle,\doteqdot \}$ and we
work with $B = \{ b_y,b_z\}$,
$J =  \{\hatzero,\hatone\} $, and a unique (up to labeling of the ports)
multicontext $m$, with 
$Y = \{ y_0,y_1 \} = \mu^{-1}(b_y)$ and $Z = \{z\}= \mu^{-1}(b_z)$.
The proof consists in circuits that are pairs
$\Brvpsi(p_\theta^{(k)}) ,\Brvpsi(p_{\theta+1}^{(k)})$,
for every $\sigma \ge \tau$ and  appropriate integers $\theta $
and $k $,
such that
$\Brvpsi(p_\theta^{(k)})\, \crochet{\star \calJ^{\sigma,1}}\, \Brvpsi(p_{\theta+1}^{(k)})$,
$\{ \varphi  (p_\theta^{(k)},\psi) = \hatzero$
if $\theta$ is odd, and
$\{ \varphi  (p_\theta^{(k)},\psi) = \hatone$
if $\theta$ is even.
The first step in the construction of 
$p_\theta^{(k)}$  and  $p_{\theta+1}^{(k)}$
consists in pumping  $m$ along $Z$
to create $ m^{(\theta,Z)}$ and $ m^{(\theta+1,Z)}$, 
then inserting  $\hatone$
at the end of the \aaa-paths, and denoting 
the resulting multicontexts  $p_\theta^{(1)} $
and  $p_{\theta+1}^{(1)}$.
An induction on $\theta$ shows that a suitable $\psi$ 
exists for every $\theta$,
for which  $\varphi(p_\theta^{(1)},\psi) \neq \bot $. 
%
\newline
For every $k \ge 1$, we build
multicontexts $p_\theta^{(k+1)}$  and  $p_{\theta+1}^{(k+1)}$
by inserting
copies of $p_\theta^{(1)} $ and  $p_{\theta+1}^{(1)} $
at the ports of $Y(p_\theta^{(k)} )$ and  $Y(p_{\theta+1}^{(k)} )$:
if the port $y$ satisfies $\psi(y) = \hatzero$,
then $p_\theta^{(1)} $ is inserted if $\theta$ is odd (and
$\varphi(p_\theta^{(1)},\psi) = \hatzero $),
otherwise this is  $p_{\theta+1}^{(1)} $.
\newline
We define a coordinate system for
every node and port of $p_\theta^{(k)}$, $k \ge 1$.
With the convention that the root, a
`$\vartriangle$' node, sits at depth level $1$,
we give to the  `$\vartriangle$' node at depth level $d$
in $p_\theta^{(1)} $
and to its `$\doteqdot$' son 
the coordinates
$a_d^{(\theta)}$ and $b_d^{(\theta)}$, respectively.
If we work on
$p_\theta^{(1)} $, i.e. no multicontext is inserted at the ports of
$Y(p_\theta^{(1)}  )$, then
the sons of  $b_d^{(\theta)}$ 
are the ports $y_{d0}^{(\theta)}$ and $y_{d1}^{(\theta)}$.
Otherwise, the two sons of
$b_d^{(\theta)}$ are `$\vartriangle$' nodes; the one
inserted at $y_{di}$
receives the coordinate $a_{di1}^{(\theta)}$,
 its `$\doteqdot$' son  $b_{di1}^{(\theta)}$,
and the two ports below it  
$y_{di10}^{(\theta)}$ and $y_{di11}^{(\theta)}$.
Along the \aaa-path from  $a_{di1}^{(\theta)}$,
the `$\vartriangle$' node at depth $e$
receives coordinate $a_{die}^{(\theta)}$,
its `$\doteqdot$' son the coordinate $b_{die}^{(\theta)}$,
and the ports below it  
$y_{die0}^{(\theta)}$ and $y_{die1}^{(\theta)}$.
In $p_\theta^{(3)}$, there ports are
replaced with `$\vartriangle$' nodes with coordinates
$a_{die01}^{(\theta)}$ and $a_{die11}^{(\theta)}$, etc.
\newline
We now look at two tuples
$(p_\theta^{(k)},\psi)$ and $(p_{\theta+1}^{(k)},\psi) $;
then we observe that 
at every depth $d \le \theta$ along the \aaa-paths of the roots, exactly one of   
$a_d^{(\theta)}$  and  $a_d^{(\theta+1)}$ 
is an $a\sfone$ node; the other is an $a\sfzero$.
When  $a_d^{(\theta)}$ is an $a\sfone$, then its son 
$b_d^{(\theta)}$ feeds it a \hatzero, and this occurs exactly when
$\psi(y_{d0}^{(\theta)}) \neq \psi(y_{d1}^{(\theta)})$
if $k=1$, and when only one of
$a_{d01}^{(\theta)}$ and $a_{d11}^{(\theta)}$
is an $a\sfone$ node, if $k \ge 2$.
This is the pattern denoted $Q_0$ in Figure 5.
Otherwise, $a_d^{(\theta)}$ is a $a\sfzero$ node, 
$b_d^{(\theta)}$ feeds it a \hatone, and either $k=1$ and
$y_{d0}^{(\theta)}$ and $y_{d1}^{(\theta)}$ carry the same label,
or $k \ge 2$ and $a_{d01}^{(\theta)}$ and $a_{d11}^{(\theta)}$
are both $a\sfzero$ or both $a\sfone$ nodes; these
are the patterns $Q_{10}$ and $Q_{11}$.
%
%
Assume that $a_d^{(\theta)}$  is the root of a pattern $Q_0$: 
we count those ports labelled
with $ \hatzero$ and those with $\hatone$,
the result is a pair 
$(n_d^{(\theta,\sfzero)},n_d^{(\theta,\sfone)}) = 
(1,1) = (2^0,2^0)$. The same work done 
on a $Q_{10}$ rooted at $a_d^{(\theta)}$
gives
$(n_d^{(\theta+1,\sfzero)},n_d^{(\theta+1,\sfone)}) = (2,0) = (2^0+1,2^0-1)$;
if this is a $Q_{11}$,  
then
$(n_d^{(\theta+1,\sfzero)},n_d^{(\theta+1,\sfone)}) = (0,2) = (2^0-1,2^0+1)$, respectively.
Observe that 
the numbers
$n_d^{(\theta,\sfzero)}$ and $n_d^{(\theta+1,\sfzero)}$,
as well as
$n_d^{(\theta,\sfone)}$ and $n_d^{(\theta+1,\sfone)}$,
differ by $1$.
%
%
\begin{figure}
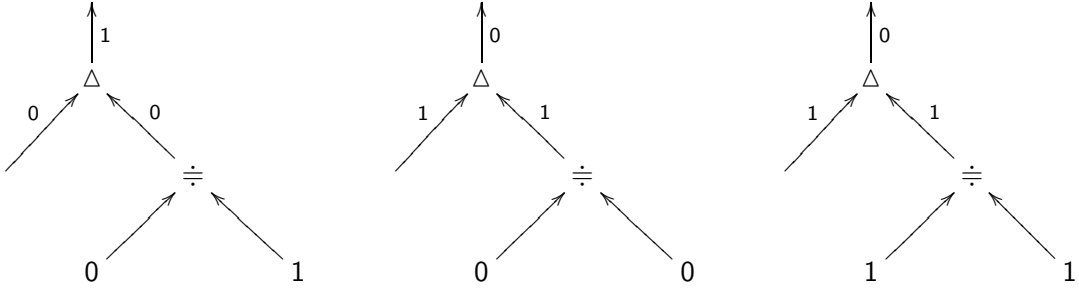

$$\diagram
&   & & & &  & & & &  & \\
& \vartriangle  \uto_\sfone & & \quad & & \vartriangle \uto_\sfzero & 
& \quad & & \vartriangle \uto_\sfzero & \\
\urto^\sfzero  &    & \doteqdot \ulto_\sfzero  &  & \urto^\sfone & 
 & \doteqdot \ulto_\sfone &  & \urto^\sfone & & \doteqdot \ulto_\sfone &  \\
  & \sfzero \urto & & \sfone \ulto & & \sfzero\urto & & \sfzero\ulto & & \sfone\urto & & \sfone\ulto \\
\enddiagram$$

\caption{patterns $Q_0$ (left), $Q_{10}$ (center) and $Q_{11}$ (right).}
\end{figure}
\newline
Assume that
$a_d^{(\theta)}$ is an $a\sfone$ node;
then for every depth $e$, exactly
one of $a_{d0e}^{(\theta)}$ and
$a_{d1e}^{(\theta)}$ is an $a\sfone$ node
and the root of  a $Q_0$ pattern;
the other is the root of a $Q_{10}$ or a $Q_{11}$.
When $k=2$, the leaves of these patterns are 
the four ports $y_{d0e0}^{(\theta)}$, $y_{d0e1}^{(\theta)}$, $y_{d1e0}^{(\theta)}$, 
and $y_{d1e1}^{(\theta)}$;
their contexts within  $p_\theta^{(2)}$ are equivalent under $\stackrel{\sigma,1}{\leftrightarrow}$,
so that it makes sense to 
look at the pair of numbers obtained by counting the
\hatzero s and  \hatone s at these   ports; the result
is
\[
(n_{d0e}^{(\theta,2,\sfzero)},n_{d1e}^{(\theta,2,\sfone)}) \in \{ (3,1), (1,3) \}
=  \{ (2^1+1,2^1-1) , (2^1-1,2^1+1) \} .
\]
Meanwhile,
$a_d^{(\theta+1)}$ is an $a\sfzero$ node, so that
$\psi(y_{d0}^{(\theta+1)}) = \psi(y_{d1}^{(\theta+1)})$
and at every depth $e$,  $a_{d0e}^{(\theta+1)}$
and
$a_{d1e}^{(\theta+1)}$ are both
$a\sfzero$ nodes or both
$a\sfone$ nodes.
In the former case, the ports under  
their `$\doteqdot$' sons constitute two $Q_0$ patterns;
otherwise, we can place a $Q_{10}$ under one of them
and a $Q_{11}$ under the other. Either way, we obtain
$n_{d0e}^{(\theta+1,2,\sfzero)} = n_{d1e}^{(\theta+1,2,\sfone)} =  (2,2) $.
\newline
Iterating the process and building 
multicontexts $p_\theta^{(k)}$ and $p_{\theta+1}^{(k)}$,
it is possible to redefine 
$\psi$ on $Y(p_\theta^{(k)}  ) $ and $Y(p_{\theta+1}^{(k)}  ) $
in such a way that for every set of $2^k$ ports
whose contexts are equivalent under $\stackrel{\sigma,1}{\leftrightarrow}$,
the corresponding  pairs of numbers 
$(n_{d\ldots}^{(\theta,k,\sfzero)} , n_{d\ldots}^{(\theta,k,\sfone)}) $ 
and
$(n_{d\ldots}^{(\theta+1,k,\sfzero)} , n_{d\ldots}^{(\theta+1,k,\sfone)}) $ 
are 
$(2^{k-1},2^{k-1})$ 
below the $a\sfzero$ node,
and one of
$(2^{k-1}+1,2^{k-1}-1) $ and $(2^{k-1}-1,2^{k-1}+1)$ 
below the $a\sfone$ node.
\newline
Consider two ports $y_{di}^{(\theta)}$ 
within $p_\theta^{(1)}$ and $y_{d'j}^{(\theta+1)}$ within
$p_{\theta+1}^{(1)}$, where $i,j \in \{0,1\}$: we have
\[
\nabla(p_\theta^{(1)} ,y_{di}^{(\theta)}) \leftrightarrow_{\sigma,1} \nabla(p_{\theta+1}^{(1)} ,y_{d'j}^{(\theta+1)}) 
\quad \text{iff} \quad
d \equiv_{\sigma,1} d'
\ \  \text{and} \ \ 
\theta-d \equiv_{\sigma,1} \theta+1-d' .
\]
If $\theta \ge 3\sigma$, then for every pair $d,i$,
the equivalence class of $\nabla(p_\theta^{(1)} ,y_{di}^{(\theta)})$   contains 
the contexts of
at least four ports of $p_\theta^{(1)}$ and four of $p_{\theta+1}^{(1)}$, 
and at least $4\sigma$ in the case where $\sigma < d < \theta-\sigma$.
In every case, for each class the numbers of ports labelled \hatzero\
and of ports labelled \hatone\ are both at least $0 = 2^0-1$.
Applying the same reasoning  to $p_\theta^{(k)}$
and $p_{\theta+1}^{(k)}$, we find that every equivalence class
that contains the context of a port of $Y( p_\theta^{(k)} )$
contains at least $2^k$ such contexts, and the numbers of
\hatzero s and \hatone s are both at least $2^{k-1}-1$.
For every threshold $\tau$ and iteration level
$n $, one can take integers $\sigma$ and $k$ such that
$\leftrightarrow_{\sigma,1}$ refines
$\approx^{n+1}_{\tau,1}$, and 
$k \ge \log_2 \sigma$,
build the multicontexts $p_\theta^{(k)}$
and $p_{\theta+1}^{(k)}$: by the
above reasoning, a labelling $\psi$
can be defined on them, such that 
$(p_\theta^{(k)},\psi) \, \crochet{\star \calJ^{\sigma,1}} \, (p_{\theta+1}^{(k)},\psi)$.
These tuples constitute a circuit $\calM^{({n+1})}$
that satisfies
$ \mathbf{RC}(\calG {\star} \calH^n_{\tau,1} )$.
Therefore, the Potthoff algebra does not belong to any
variety
${\mathbf{**}} \bicrochet{\N_{\tau,1}}$.
%




\section{Conclusion}   \label{sec:conclusion}

This research started with the intention of understanding
the mechanisms that underlie the Ehrenfeucht-Fra\"{\i}ss\'e games
used to prove that a forest algebra lies outside of
${\mathbf{**}}\bicrochet{\N_{\tau,\pi} }$.
Multicontexts and tuples, where ports have multiple labelings,
are at the center of this proof technique,
and  the need for a description of how
a forest algebra works on these objects led to the
definition of the algebras $\calG_\#$ and $\calG_\%$.
In parallel to this came the observation that
Ehrenfeucht-Fra\"{\i}ss\'e games
are recursive and uniform; while games with these two properties  always
exist on monoids, word languages and
linear orderings, it is not clear whether the same holds
in the world of trees and forests.
We could prove the existence of a recursive proof for every
algebra outside of 
${\mathbf{**}}\bicrochet{\N_{\tau,\pi} }$,
while our investigation of proof-by-pumping suggests 
has not solved 
the question of the existence of a uniform proof.
\newline
Given an algebra $\calG = (G,W)$, with  $\sigma$ and $\rho$
the threshold and period of $\calG_\%$, Theorem \ref{thm:RC}\
can be combined with Propositions \ref{prp:firstSD}\ and \ref{prp:Jsigma}\ 
into this statement:
\begin{equation}  \label{eqn:IN}
\text{if}\
\calG \in {\mathbf{**}}\bicrochet{\N_{\tau,\pi} }, \
\text{then}\   
\MVM W_\% \in \mathbf{Sol}_{\tau,\pi} \
\text{and}\   
\mathbf{NRC}(\calG  {\star}   \calJ^{\sigma,\rho})
\end{equation}
or conversely,
\begin{equation}  \label{eqn:OUT}
\text{if}\ \
\MVM W_\% \not\in \mathbf{Sol}_{\tau,\pi} \
\text{or}\   
\mathbf{RC}(\calG {\star}  \calJ^{\sigma,\rho}), \
\text{then}\   
\calG \not\in {\mathbf{**}}\bicrochet{\N_{\tau,\pi} } .
\end{equation}
All examples of algebras outside of $ {\mathbf{**}}\bicrochet{\N_{\tau,\pi} }$
known to the author satisfy the left hand side of 
(\ref{eqn:OUT}), so that it makes sense to ask
whether 
these implications are actually equivalences.
If they are, then since
satisfaction of
$\mathbf{RC}(\calG  {\star}   \calJ^{\sigma,\rho})$
is a recursively enumerable problem,
membership of $\calG $ in 
${\mathbf{**}}\bicrochet{\N_{\tau,\pi} }$
is decidable, 
by running in parallel
a search of an $n \ge 1$ such that $\calG \prec \calH^n_{\tau,\pi}$,
and a test for non-membership that consists in verifying the left hand side
of (\ref{eqn:OUT}),
knowing that one of them will stop.
Proving that
$\mathbf{RC}(\calG  {\star}   \calJ^{\sigma,\rho})$
is by itself decidable would of course yield
a procedure more elegant than this.
\newline
Satisfaction of $\mathbf{RC}(\calG {\star}  \calJ^{\sigma,\rho})$
implies the existence of a recursive proof for
${\mathbf{**}}\bicrochet{\N_{\tau,\pi} }$;
whether it also implies the 
existence of a uniform one
is an open question. 
If there exist algebras
outside of ${\mathbf{**}}\bicrochet{\N_{\tau,\pi} }$
that have only non-uniform proofs,
then one
would ask whether this
property is decidable.
\newline
It is also worth noting that
conversely, Theorem \ref{thm:RC}\ suggests 
that if  membership in
${\mathbf{**}}\bicrochet{\N_{\tau,\pi} }$
is undecidable, then this
might be proved by showing that
the existence of a recursive non-membership proof 
is undecidable.
\newline
Condition $\mathbf{RC}(\calG  {\star}  \calH^n _{\tau,\pi})$ gives a  new viewpoint on the combinatorics of
first-order languages; it involves counting under a threshold and a period,
just as in the one developed in
\cite{hath87,mora03},
but  in a 
nonlocal (and messy) way along an antichain of ports.
Whether they can be deduced from one another remains to be seen.
\vspace{0.2in}
\newline
The author  thanks  Howard Straubing, both  for  discussions 
about this research and for giving him access to unpublished material. 
Discussions with Andreas Krebs and Charles Paperman were also useful.

%

\end{document}

%% file: MWrite.tex
\newcommand{\jbox}{\hspace*{\fill} $\Box$}

\newcommand{\Z}{\ensuremath{\mathbb Z}}

\newcommand{\N}{\ensuremath{\mathbb N}}

\newcommand{\aaa}{\mbox{\sf a}}

\newcommand{\sftwo}{\ensuremath{\mathsf{2}}}
\newcommand{\sfone}{\ensuremath{\mathsf{1}}}
\newcommand{\sfzero}{\ensuremath{\mathsf{0}}}

\newcommand{\crochet}[1]{\langle #1 \rangle}


%% file: MEnvA.tex
\newtheorem{theorem}{Theorem}[section]

\newtheorem{lemma}[theorem]{Lemma}

\newtheorem{proposition}[theorem]{Proposition}
\newtheorem{definition}[theorem]{Definition}

\newenvironment{proof}{{\sl Proof.}}{\jbox}

\newenvironment{listedense}{\begin{list}{$\bullet$}{\setlength{\parsep}{0pt}
	\setlength{\parskip}{0pt}
	\setlength{\topsep} {0pt} \setlength{\itemsep} {0pt}
	\setlength{\labelsep}{1em}}}{\end{list}}

\newcounter{zahl}

\newcounter{latinzahl}
\newenvironment{latindense}{\begin{list}{$\roman{latinzahl}.$}{\usecounter{latinzahl}
	\setlength{\parsep}{0pt}
	\setlength{\parskip}{0pt}
	\setlength{\topsep} {0pt} \setlength{\itemsep} {0pt}
	\setlength{\labelsep}{1em}}}{\end{list}}

\newcounter{buchstabe}

%% file: MBibli.tex
\newcommand{\aut}[1]{{\sc #1}}

\newcommand{\ico}[3]{{\em Information and Control\ }{\bf #1} (#2), pp. #3.}
\newcommand{\ic}[3]{{\em Information and Computation\ }{\bf #1} (#2), pp. #3.}
\newcommand{\ijac}[3]{{\em International Journal of Algebra  and Computation\ }{\bf #1} (#2), pp. #3.}

\newcommand{\jcss}[3]{{\em J. Computer and Systems Science\ }{\bf #1} (#2), pp. #3.}

\newcommand{\jpaa}[3]{{\em J. of Pure and Applied Algebra\ }{\bf #1} (#2), pp. #3.}
\newcommand{\lmics}[4]{{\em Logical Methods in Computer Science\ }{\bf #1} #2 (#3), pp. #4.}

\newcommand{\picalp}[4]{{\em Proc. of the #1th International Colloquium on Automata, Languages and Programming}, {\em Lecture Notes in Comp. Sci.} \textbf{#2}, Springer-Verlag (#3), pp. #4.}

\newcommand{\plics}[3]{{\em Proc. of the #1th IEEE Symp. on Logics in Computer Science} (#2), pp. #3.}

\newcommand{\rairo}[3]{{\em Revue d'Automatique, d'Informatique et de Recherche Op\'erationnelle\ }{\bf #1} (#2), pp. #3.}

\newcommand{\tams}[3]{{\em Trans. Amer. Math. Soc. }{\bf #1} (#2), pp. #3.}

\newcommand{\tcs}[3]{{\em Theoretical Computer Science\ }{\bf #1} (#2), pp. #3.}

%% file: fmacros_a.tex
\newcommand{\bicrochet}[1]{\langle \! \langle #1 \rangle \! \rangle}
\newcommand{\bitruc}[1]{  [[ #1 ]] }

\newcommand{\MVM}{\ensuremath{{}^{\underline{m}}}}

\newcommand{\hatzero}{\ensuremath{\hat{\Zero}}}
\newcommand{\hatone}{\ensuremath{\hat{\One}}}

\newcommand{\Brvmu}{\ensuremath{\breve{\mu}}}

\newcommand{\Brvxi}{\ensuremath{\breve{\xi}}}
\newcommand{\Brvphi}{\ensuremath{\breve{\varphi}}}
\newcommand{\Brvchi}{\ensuremath{\breve{\chi}}}
\newcommand{\Brvpsi}{\ensuremath{\breve{\psi}}}

\newcommand{\calB}{\ensuremath{\mathcal{B}}}

\newcommand{\calD}{\ensuremath{\mathcal{D}}}

\newcommand{\calG}{\ensuremath{\mathcal{G}}}
\newcommand{\calH}{\ensuremath{\mathcal{H}}}
\newcommand{\calI}{\ensuremath{\mathcal{I}}}
\newcommand{\calJ}{\ensuremath{\mathcal{J}}}
\newcommand{\calK}{\ensuremath{\mathcal{K}}}
\newcommand{\calL}{\ensuremath{\mathcal{L}}}
\newcommand{\calM}{\ensuremath{\mathcal{M}}}

\newcommand{\calP}{\ensuremath{\mathcal{P}}}

\newcommand{\calR}{\ensuremath{\mathcal{R}}}
\newcommand{\calS}{\ensuremath{\mathcal{S}}}
\newcommand{\calT}{\ensuremath{\mathcal{T}}}
\newcommand{\calU}{\ensuremath{\mathcal{U}}}

\newcommand{\bfW}{\ensuremath{\mathbf{W}}}

\newcommand{\Bfe}{\ensuremath{\mathbf{e}}}
\newcommand{\Bfo}{\ensuremath{\mathbf{o}}}

\newcommand{\sfW}{\ensuremath{\mathsf{W}}}

\newcommand{\BbbC}{\ensuremath{\mathbb{C}}} 
  
\newcommand{\BbbF}{\ensuremath{\mathbb{F}}} 
\newcommand{\BbbH}{\ensuremath{\mathbb{H}}} 
 
\newcommand{\BbbM}{\ensuremath{\mathbb{M}}}

\newcommand{\BbbV}{\ensuremath{\mathbb{V}}}

\newcommand{\Gin}{\ensuremath{G_{\text{in}}}}  
\newcommand{\Gout}{\ensuremath{G_{\text{out}}}}

\newcommand{\calHin}{\ensuremath{\mathcal{H}_{\text{in}}}}  

\newcommand{\One}{\ensuremath{\mathtt{1}}}
\newcommand{\Zero}{\ensuremath{\mathtt{0}}}

\newcommand{\Panc}{\ensuremath{\prec}} 

\newcommand{\Tree}{\ensuremath{\text{Tree}}}

\newcommand{\Blockprod}{\ensuremath{\; {\scriptstyle{\square}} \; }}

\newcommand{\ODalg}{\ensuremath{\mathcal{O}\!\!\!\mathcal{D}}}

\newcommand{\Ltree}{\ensuremath{L^{\mathsf{tree}}}}

\newcommand{\Lmin}{\ensuremath{L^{\mathsf{min}}}}

\newcommand{\Sstar}{\ensuremath{\mathbb{T}_A}}
 \newcommand{\Snat}{\ensuremath{\mathbb{H}_A}} 
  \newcommand{\Sfree}{\ensuremath{\mathbb{F}_A}}

       \newcommand{\Svert}{\ensuremath{\mathbb{V}_A}} 
   \newcommand{\Sboxvert}{\ensuremath{\mathbb{V}_A}}